\documentclass[twocolumn,fleqn,10pt]{wlscirep}
\usepackage[utf8]{inputenc}
\usepackage[T1]{fontenc}
\usepackage{color}
\usepackage{multirow}
\usepackage{pifont}
\usepackage[ruled,vlined]{algorithm2e}
\usepackage{pdfpages}
\SetKwInput{KwInput}{Input}
\SetKwInput{KwOutput}{Output}
\usepackage{amsthm}
\usepackage{mathrsfs}
\newcommand{\raggedr}{\leftskip=0pt \rightskip=0pt plus 0cm}

\newcommand{\revision}[1]{\textcolor[rgb]{0.00,0.00,0.00}{#1}}

\newtheorem{theorem}{Theorem}

\makeatletter
\newcommand{\removelatexerror}{\let\@latex@error\@gobble}
\makeatother

\def\framework{QuantumFlow}
\def\qfsnet{QF-Net}
\def\qfnet{QF-pNet}
\def\qfhnet{QF-hNe}
\def\qfsim{QF-FB}
\def\qfcirc{QF-Circ}
\def\qfmap{QF-Map}
\def\qfnc{-LY}
\def\qfbn{N-LYR}

\title{\vspace{25pt}A Co-Design Framework of Neural Networks and Quantum Circuits Towards Quantum Advantage}

\author[1]{Weiwen Jiang}
\author[2]{Jinjun Xiong}
\author[1]{Yiyu Shi}
\affil[1]{University of Notre Dame, Notre Dame, IN, 46556, USA}
\affil[2]{IBM Thomas J. Watson Research Center, Yorktown Heights, NY, 10598, USA}

\begin{abstract}

Despite the pursuit of quantum advantages in various applications, the power of quantum computers in machine learning (such as neural network models) has mostly remained unknown, primarily due to a missing link that effectively designs a neural network model suitable for quantum circuit implementation.
In this article, we present the first co-design framework, namely \framework, \revision{to provide such a} missing link. \framework~consists of novel quantum-friendly neural networks (\qfsnet{s}), an automatic mapping tool (\qfmap) to generate the quantum circuit (\qfcirc) for \qfsnet{s}, and an execution engine (\qfsim)
to efficiently support the training of \qfsnet{s} on a classical computer.
We discover that, in order to make full use of the strength of quantum representation, \revision{it is best to represent data in a neural network as either random variables or numbers in unitary matrices, such that they can be directly operated by the basic quantum logical gates.
Based on these data representations, we propose two quantum friendly neural networks, \qfnet~and \qfhnet{t} in \framework.
\qfnet~using random variables (i.e., the probabilistic model) has better flexibility, and can seamlessly connect two layers without measurement with more qbits and logical gates
than \qfhnet{t}.
On the other hand, \qfhnet{t} with unitary matrices  can encode $2^k$ data into $k$ qbits, and a novel algorithm can guarantee the cost complexity (i.e., logical gates) to be $O(k^2)$.
Compared to the cost of $O(2^k)$ in classical computing and the existing quantum implementations, \qfhnet{t} demonstrates the quantum advantages.} 
Evaluation results show that \qfnet~and \qfhnet{t} can achieve 97.10\% and 98.27\% accuracy, respectively, in distinguishing digits 3 and 6 in the widely used MNIST dataset, which are 14.55\% and
15.72\% higher than the state-of-the-art quantum-aware implementation.
\revision{Results further show that for input sizes of neural computation grow from 16 to 2,048, the cost reduction of \framework~increased from 2.4$\times$ to 64$\times$.
Furthermore, on MNIST dataset, \qfhnet{t} can achieve accuracy of 94.09\%, while the cost reduction against the classical computer reaches 10.85$\times$.}
Finally, a case study on a binary classification application is conducted. Running on IBM Quantum processor's ``ibmq\_essex'' backend, a neural network designed by \framework~can achieve 82\% accuracy.
To the best of our knowledge, \framework~is the first framework that co-designs the neural networks and quantum circuits\revision{, and the first work to demonstrate the potential quantum advantage on neural network computation.}

\end{abstract}
\begin{document}

\setlength{\textfloatsep}{7pt}
\setlength{\floatsep}{7pt}
\setlength{\dbltextfloatsep}{7pt}
\setlength{\abovecaptionskip}{5pt}

\flushbottom
\maketitle
% \thispagestyle{empty}

% \noindent Please note: Abbreviations should be introduced at the first mention in the main text – no abbreviations lists. Suggested structure of main text (not enforced) is provided below.

\section*{Introduction}
% \revision{Although quantum computers are expected to dramatically outperform the classical computer, so far, quantum advantages have been shown in the limited number of applications, such as prime factorization\cite{shor1999polynomial} and sampling the output of random quantum circuits\cite{arute2019quantum}.
% In this work, we will demonstrate that quantum computers can achieve potential quantum advantage on the most commonly used neural network computation.}

\revision{Although quantum computers are expected to dramatically outperform classical computers, so far quantum advantages have only been shown in a limited number of applications, such as prime factorization\cite{shor1999polynomial} and sampling the output of random quantum circuits\cite{arute2019quantum}. In this work, we will demonstrate that quantum computers can achieve potential quantum advantage on neural network computation, a very common task in the prevalence of artificial intelligence (AI)\footnote{Quirk demos at \url{https://wjiang.nd.edu/categories/qf/}}.}

In the past decade, neural networks \cite{lecun2015deep,goodfellow2016deep,szegedy2015going} have become the mainstream machine learning models, and have achieved consistent success in numerous Artificial Intelligence (AI) applications, such as image classification \cite{krizhevsky2012imagenet,he2016deep,simonyan2014very,szegedy2016rethinking}, object detection \cite{lin2017feature,ren2015faster,he2017mask,ronneberger2015u}, and natural language processing \cite{young2018recent,sak2014long,vaswani2017attention}.
\revision{When the neural networks are applied to a specific field (e.g., AI in medical or AI in astronomy), the high-resolution input images bring new challenges. For example, one 3D-MRI image contains $224\times224\times10\approx5\times10^6$ pixels\cite{bernard2018deep} while one Square Kilometre Array (SKA) science data contains $32,768\times 32,768\approx1\times10^9$ pixels\cite{bonaldi2018square,lukic2020convosource}.
The large inputs greatly increase the computation in neural network\cite{xu2018scaling}, which gradually becomes the performance bottleneck.}
% The key factor is the significantly improved prediction accuracy.
% by making networks deeper, known as deep learning; however, with the growing depth of neural networks, the storage and computation requirement sharply increases \cite{xu2018scaling}, which gradually becomes the performance bottleneck in classical computers (e.g., the well-known memory wall issues \cite{wulf1995hitting}).
Among all computing platforms, the quantum computer is \revision{a} most promising one to address such challenges \cite{arute2019quantum,steffen2011quantum} as a quantum accelerator for neural networks \cite{schuld2015introduction,bertels2019quantum,cai2015entanglement}.
Unlike classical computers with $N$ digit bits to represent $1$ N-bit number at one time, quantum computers with $K$ qbits can represent $2^K$ 
% \revision{K-bit} 
numbers and manipulate them at the same time \cite{nielsen2002quantum}.
Recently, a quantum machine learning programming framework, TensorFlow Quantum, has been proposed \cite{broughton2020tensorflow}; however, how to exploit the power of quantum computing for neural networks is still remained unknown.

% \begin{figure*}[t]
% \centering
% \includegraphics[width=\linewidth]{Figure/overview_qflow.eps}
% \caption{\raggedr{QuantumFlow, an end-to-end neural network and quantum circuit co-design: (a-b) Machine Learning World: image classification on MNIST dataset; data representation using random variables followed by two-point distribution; quantum-friendly neural network model, \qfnet; (c) Quantum World: the innovative encoder circuit and batch normalization unit in \qfcirc; (d) Make it flow: an efficient simulator \qfsim~based on the fundamental probability theory. (Best viewed in color)}}
% \label{fig:stream}
% \end{figure*}

\begin{figure}[t]
\centering
\includegraphics[width=0.99\linewidth]{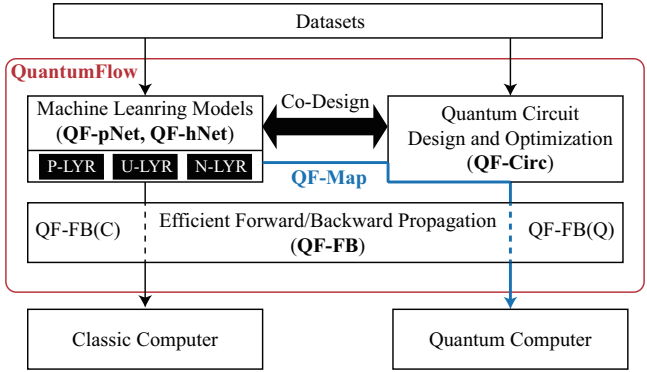}
\caption{\raggedr{\revision{\framework, an end-to-end co-design framework, provides a missing link between neural network and quantum circuit designs, which is composed of \qfsnet{s}, \qfhnet{t}, \qfsim, \qfcirc, \qfmap~that work collaboratively to design neural networks and their quantum implementations.}}}
% ; (b) a data representation of data sample (e.g., from the MNIST dataset) using random variables following a two-point distribution; (c) a quantum-friendly neural network with batch normalization.}}}
\label{fig:stream}
\end{figure}

One of the most challenging obstacles to implementing neural network computation on a quantum computer is the missing link between the design of neural networks and that of quantum circuits. 
The existing works separately design them from two directions.
The first direction is to map the existing neural networks designed for classical computers to quantum circuits;
for instance, recent works\cite{tacchino2019artificial,tacchino2019quantum,rebentrost2018quantum,schuld2014quest} map McCulloch-Pitts (MCP) neurons \cite{mcculloch1943logical} onto quantum circuits.
Such an approach 
% can take full use of the traditional innovations in machine learning (e.g., the stochastic gradient descent in training model), but 
has difficulties in consistently mapping the trained model to quantum circuits.
For example, it needs a large number of qbits to realize the multiplication of real numbers.
% To make the mapping consistent,
To overcome this problem, some existing works\cite{tacchino2019artificial,tacchino2019quantum,rebentrost2018quantum,schuld2014quest} assume binary representation (i.e., ``-1'' and ``+1'') of activation, which cannot well represent data as seen in modern machine learning applications. 
% For instance, in computer vision related applications, data in images are commonly represented as real numbers.
\revision{This has also been demonstrated in work\cite{havlivcek2019supervised}, where data in the interval of $(0, 2\pi]$ instead of binary representation are mapped onto the Bloch sphere to achieve high accuracy for support vector machines (SVMs).}
In addition, some typical operations in neural networks cannot be implemented on quantum circuits, leading to inconsistency.
For example, to enable deep learning, batch normalization is a key step in a deep neural network to improve the training speed, model performance, and stability; however, 
% in the existing multi-layer network implementation\cite{tacchino2019quantum}, the batch normalization is not applied since 
directly conducting normalization on the output qbit (say normalizing the qbit with maximum probability to probability of 100\%) is equivalent to reset a qbit without measurement, which is simply impossible.
In consequence, batch normalization is not applied in the existing multi-layer network implementation\cite{tacchino2019quantum}.

The other direction is to design neural networks dedicated to quantum computers, like the tree tensor network (TTN)\cite{shi2006classical,grant2018hierarchical}.
Such an approach 
% has the potential to exploit quantum advantages but 
suffers from scalability problems.
More specifically, 
% it lacks the efficient forward/backward propagation procedure on classical computers, leading to a 2-layer neural network to take hundreds of CPU days to train a model.
% As such, it is simply intolerant for training a larger network, which limits the scale of networks.
% , which is obviously not scalable.
the effectiveness of neural networks is based on a trained model via the forward and backward propagation on large training sets.
However, it is too costly to directly train one network by applying thousands of times forward and backward propagation on quantum computers; in particular, there are limited available quantum computers for public access at the current stage.
An alternative way is to run a quantum simulator on a classical computer to train models, but the time complexity of quantum simulation is $O(2^m)$, where $m$ is the number of qbits.
This significantly restricts the trainable network size for quantum circuits.

To address all the above obstacles, it demands to take quantum circuit implementation into consideration when designing neural networks.
This paper proposes the first co-design framework, namely \framework, where \revision{five sub-components (\qfnet, \qfhnet{t}, \qfsim, \qfcirc, and \qfmap) work collaboratively to design neural networks and implement them to quantum computers, as shown in Figure \ref{fig:stream}}.

% \framework~contains four sub-components: \qfnet, \qfsim, \qfcirc, and \qfmap.
% These components work collaboratively to design a neural network and implement it to a quantum computer.

\revision{In \framework, the start point is the co-design of networks and quantum circuits.
We first propose \qfnet, which contains multiple neural computation layer, namely {P}\qfnc{R}.
In the design of {P}\qfnc{R}, we take full advantage of the ability of quantum logic gates to operate random variables represented by qbits.
Specifically, data in {P}\qfnc{R} are modeled as random variables following a two-point distribution, which is consistent to the expression of a qbit; computations in {P}\qfnc{R} can be easily implemented by the basic quantum logic gates.
% are implemented by carried out by quantum gates on input qbits
% We design the basic operation, neural computation (denoted as {P}\qfnc{R}), based on the converted random variables.
% When it comes to quantum implementation, {P}\qfnc{R}~
Kindly note that {P}\qfnc{R}~can model both inputs and weights to be random variables.
But because binary weights can achieve comparable high accuracy for deep neural network applications \cite{courbariaux2015binaryconnect} and significantly reduce circuit complexity, we employ random variables for inputs only and binary values for weights in {P}\qfnc{R}.
Benefiting from the quantum-aware data interpretation for inputs, {P}\qfnc{R}~can be attached to the output qbits of previous layers without measurement; however, it utilizes $2^k$ qbits to represent $2^k$ input data items, and the computation needs at least one quantum gate for each qbit.
Therefore, it has high cost complexity.
}

% , which can fully represent all probabilities in joint distribution and thus have the flexibility to realize different functions, such as Quadratic and ReLU non-linear functions, but it has high cost complexity.}

\revision{Towards achieving the quantum advantage, we propose a hybrid network, namely \qfhnet{t}, which is composed of two types of neural computation layers: {P}\qfnc{R} and
{U}\qfnc{R}.
{U}\qfnc{R} is based on the unitary matrix, where $2^k$ input data are converted to a vector in the unitary matrix, such that all data can be represented by the amplitudes of states in a quantum circuit with $k$ qbits.
The reduction in input qbits provides the possibility to achieve quantum advantage; however, the state-of-the-art implementation\cite{tacchino2019artificial} using hypergraph state for computation still has the cost complexity of $O(2^k)$.
In this work, we devise a novel optimization algorithm to guarantee the cost complexity of {U}\qfnc{R} to be $O(k^2)$, which takes full use of the properties of neural networks and quantum logic gates.
% takes advantage of machine learning algorithm where the orders of input neurons can be changed.
% , and \qfmap~provides an optimization algorithm that can guarantee the cost of {U}\qfnc{R} to be $O(k^2)$.
Compared with the complexity of $O(2^k)$ on classical computing platforms, {U}\qfnc{R} demonstrates the quantum advantages of executing neural network computations.}

% observe that the quantum logic gates can directly operates on random variables rather than real numbers.

% \qfnet~is a novel quantum-friendly neural network.
% In \qfnet, we discover that to take full advantage of the quantum representation, the data in a neural network, instead of being treated as real numbers, are best modeled as random variables following a two-point distribution as shown in Figure \ref{fig:stream}(b). 
% Neural Computation (NC), one key operation in \qfnet, is designed based on such converted random variables as inputs \revision{and the binary values as weights.
% It has been demonstrated that applying binary weights can achieve high performance in DNN applications \cite{courbariaux2015binaryconnect}.}

In addition to neural computation, \qfsnet{s}~also integrates a quantum-friendly batch normalization \qfbn, which can be plugged into both \qfnet~and \qfhnet{t}.
It includes additional parameters to normalize the output of a neuron, which are tuned during the training phase.

To support both the inference and training of \qfsnet{s}, we further develop \qfsim, a forward/backward propagation engine.
% , which is 
% Finally, an efficient simulator is proposed to bridge the machine learning model and quantum circuit design.
% It is 
When \qfsim~is integrated into PyTorch to conduct inference and training of \qfsnet{s}~on classical computers, we denote it as \qfsim(C). 
% whose
% design is based on the probability theory for \qfnet. \qfsim(C) is efficient for both inference and training.
\qfsim~can also be executed on a quantum computer or a quantum simulator. 
Based on Qiskit Aer simulator, we implement \qfsim(Q) for inference with or without error models. 

%  classical computers, \qfsim(C) is based on the probability theory, and therefore can be conducted efficiently to support network training.
% On quantum computers, \qfsim(Q) is built upon Qiskit Aer simulator to  obtain inference results on quantum computer with or without error models.

% within 15 seconds for testing \qfnet~on MNIST, while the .

% as shown in Figure \ref{fig:stream}(d).

% Unlike the state-based general quantum simulator, \qfsim~is dedicated for \qfnet, and can reduce the simulation time complexity from $O(2^N)$ to $O(N^2)$. As a result, the simulation time is drastically reduced from 14,650 hours to 15 seconds.

% Basically, we will show that prove that the output of each neural is the expectation of random variable, which is the weighted sum of inputs.
% As a result, we propose a dedicate polynomial time simulator, which can dramatically reduce the elapsed time in processing input data from thousands CPU hours to several CPU seconds. 
% \framework~fully considers both machine learning and quantum computing designs.
% First, 
% Inspired by this, we propose a probabilistic neural network, called ``\qfnet''.

\begin{figure*}[t]
\centering
\includegraphics[width=1.0\linewidth]{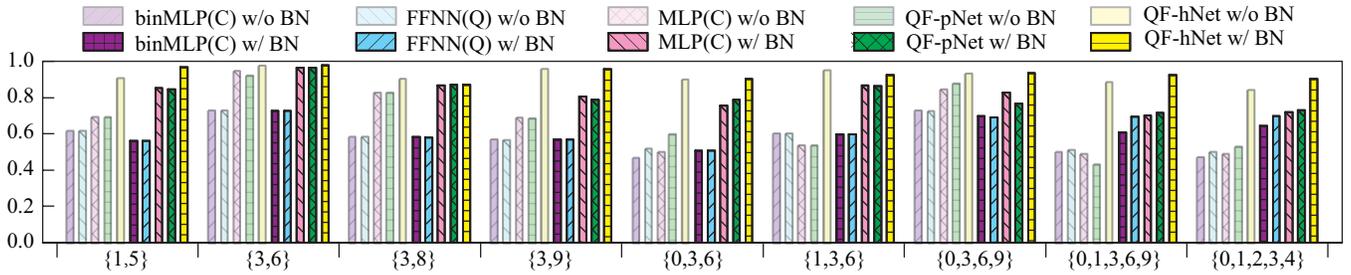}
\caption{\raggedr{\revision{\qfhnet{t}~achieves state-of-the-art accuracy in image classifications on different sub-datasets of MNIST.}}}
% For all datasets, we employ a neural network with two hidden layers, where the first hidden layer contains 4 neurons for \{3,6\}/\{3,8\} and 8 neurons for \{1,3,6\}; and the second hidden layer contains 2 neurons for \{3,6\}/\{3,8\} and 3 neurons for \{1,3,6\}.}}}
\label{fig:exp_model}
\end{figure*}

For each operation in \qfsnet{s}~(e.g., neural computations and batch normalization), a corresponding quantum circuit is designed in \qfcirc.
% Specifically, two basic components, neural computation and batch normalization, are composed in \qfcirc, as shown in Figure \ref{fig:stream}(c).
In neural computation, an encoder is involved to encode the inputs and weights.
The output will be sent to the batch normalization which involves additional control qbits to adjust the probability of a given qbit to be ranged from 0 to 1.
% Based on the given output and the tuned parameter from machine learning algorithm.
Based on \qfsnet{s}~and \qfcirc, \qfmap~is an automatic tool to conduct (1) network-to-circuit mapping (from \qfsnet{s}~to \qfcirc); (2) virtual-to-physic mapping (from virtual qbits in \qfcirc~to physic qbits in quantum processors).
Network-to-circuit mapping guarantees the consistency between \qfsnet{s}~and \qfcirc~with or without internal measurement; while virtual-to-physic mapping is based on Qiskit with the consideration of error rates.

% In a binary classification case study, \qfmap~achieves 14\% accuracy gain on IBM quantum processors, compared to the ones using the default mapping in Qiskit.

% \qfnet~to \qfcirc, and 
% \qfmap\raisebox{-1pt}{\large\ding{192}} 
% Such a quantum circuit can be implemented on quantum processors with or without internal measurement.
% Finally, \qfmap~is based on Qiskit to map qbits in \qfnet~to physic qbits with the consideration of error rates.
% In the binary classification case study, \qfmap~achieves 14\% accuracy gain on IBM quantum processors, compared to the ones using the default mapping in Qiskit.
% compared with the non-optimized ones.

% higher than the non-optimized deployment 

% In deployment, with the consideration of different error rates in physic qbits, we provide mapping rules to map virtual qbits to physic qbits.
% % Benefiting from the quantum-aware neural network design, an automatic tool is developed to build up the quantum circuit, which is composed of two basic components: neural computation and batch normalization, as shown in Figure \ref{fig:stream}(c).

As a whole, given a dataset, \framework~can design and train a quantum-friendly neural network and automatically generate the corresponding quantum circuit.
The proposed co-design framework is evaluated on the IBM Qikist Aer simulator and IBM Quantum Processors.

\section*{Results}

This section presents the evaluation results of all \revision{five sub-components} in \framework.
We first evaluate the effectiveness of \revision{\qfsnet{s} (i.e., \qfnet~and \qfhnet{t})} on the commonly used MNIST dataset \cite{lecun1998gradient} for the classification task.
Then, we show the consistency between \qfsim(C) on classical computers and \qfsim(Q) on the Qiskit Aer simulator.
\revision{Next, we show that \qfmap~is a key to achieve quantum advantage.}
We finally conduct an end-to-end case study on a binary classification test case on IBM quantum processors to test \qfcirc.

\subsection*{\qfsnet{s}~Achieve High Accuracy on MNIST}
Figure \ref{fig:exp_model} reports the results of different approaches for the classification of handwritten digits on the commonly used MNIST dataset \cite{lecun1998gradient}.
Results clearly show that with the same network structure (i.e., the same number of layers and the same number of neurons in each layer), the proposed \qfhnet{t} can achieve the highest accuracy than the existing models: \textit{(i)} multi-level perceptron (MLP) with binary weights for the classical computer, denoted as MLP(C);  \textit{(ii)} MLP with binary inputs and weights designed for the classical computer, denoted as binMLP(C); and \textit{(iii)} a state-of-the-art quantum-aware neural network with binary inputs and weights \cite{tacchino2019quantum}, denoted as FFNN(Q).  
% \todo{Details of dataset and traning settings can be found in the appendix.}

\revision{Before reporting the detailed results, we first discuss the experimental setting. 
In this experiment, we extract sub-datasets from MNIST, which originally include 10 classes.
For instance, $\{3,6\}$ indicates the sub-datasets with two classes (i.e., digits 3 and 6), which are commonly used in quantum machine learning (e.g., Tensorflow-Quantum \cite{tfquantum}).
To evaluate the advantages of the proposed \qfsnet{s}, we further include more complicated sub-datasets, \{3,8\}, \{3,9\}, \{1,5\} for two classes.
In addition, we show that \qfsnet{s} can work well on larger datasets, including \{0,3,6\} and \{1,3,6\} for three classes, and \{0,3,6,9\}, \{0,1,3,6,9\}, \{0,1,2,3,4\} for four and five classes.
For the datasets with two or three classes, the original image is downsampled from the resolution of $28\times 28$ to $4\times 4$, while it is downsampled to $8\times 8$ for datasets with four or five classes.
All original images in MNIST and the downsampled images are with grey levels.
For all involved datasets, we employ a two-layer neural network, where the first layer contains 4 neurons for two-class datasets, 8 neurons for three-class datasets, and 16 neurons for four- and five-class datasets.
The second layer contains the same number of neurons as the number of classes in datasets.
Kindly note that theses architectures are manually tuned for higher accuracy, the neural architecture search (NAS) will be our future work.}

In the experiments, for each network, we have two implementations: one with batch normalization (w/ BN) and one without batch normalization (w/o BN).
Kindly note that FFNN\cite{tacchino2019quantum} does not consider batch normalization between layers. To show the benefits and generality of our newly proposed BN for improving the quantum circuits' accuracy, we add that same functionality to FFNN for comparison.
\revision{From the results in Figure \ref{fig:exp_model}, we can see that the proposed ``\qfhnet{t}~w/ BN'' (abbr. \qfhnet{t}\_BN) achieves the highest accuracy among all networks (even higher than MLP running on classical computers). 
Specifically, for the dataset of $\{3,6\}$, the accuracy of \qfhnet{t}\_BN is 98.27\%, achieving $3.01\%$ and $15.27\%$ accuracy gain against MLP(C) and FFNN(Q), respectively.
It even achieves a $1.17\%$ accuracy gain compared to \qfnet\_BN.
An interesting observation attained from this result is that with the increasing  number of classes in the dataset, \qfhnet{t}\_BN can maintain the accuracy to be larger than $90\%$, while other competitors suffer an accuracy loss.
Specifically, for dataset \{0,3,6\} (input resolution of $4\times 4$), \{0,3,6,9\} (input resolution of $8\times 8$), \{0,1,3,6,9\} (input resolution of $8\times 8$), the accuracy of \qfhnet{t}\_BN are 90.40\%, 93.63\% and 92.62\%; however, for MLP(C), these figures are 75.37\%, 82.89\%, and 70.19\%.
This is achieved by the hybrid use of two types of neural computation in \qfhnet{t} to better extract features in images.}
% Similar improvements are also achieved for \qfhnet{t}\_BN on other datasets.
% Because of the similarity of 3 and 8, \qfnet\_BN only achieves an accuracy of 86.95\%, but
% it is still the best accuracy among all networks.
% \qfnet\_BN achieves $14.55\%$ accuracy gain.
The above results validate that the proposed \qfhnet{t}~has a great potential in solving machine learning problems and our co-design framework is effective to design a quantum neural network with high accuracy.
% To the best of the authors knowledge, this is the best accuracy result obtained by a quantum-aware neural network model on MNIST.

\begin{table*}[t]
  \centering
  \scriptsize
  \tabcolsep 4.2pt
  \renewcommand\arraystretch{1.3}
  \caption{\raggedr{\revision{Inference accuracy and efficiency comparison between \qfsim(C)~and \qfsim(Q) on both \qfnet~and \qfhnet{t} using MNIST dataset to show the consistency of implementations of \qfsnet{s} on classical computers and quantum computers.}}}
    \begin{tabular}{ccccccccccccccc}
    \toprule
    \multirow{2}[4]{*}{} &  \multicolumn{7}{l}{\textbf{\qfnet}} & \multicolumn{7}{l}{\revision{\textbf{\qfhnet{t}}}}\\
    % \cmidrule{2-12}
    & \multicolumn{2}{c}{Qbits (Neurons)} & \multicolumn{3}{c}{Accuracy}  & \multicolumn{2}{c}{Elapsed CPU Time} & \multicolumn{2}{c}{\revision{Qbits (Neurons)}} & \multicolumn{3}{c}{\revision{Accuracy}}  & \multicolumn{2}{c}{\revision{Elapsed CPU Time}} \\
\cmidrule{1-15}       dataset   & L1    & L2    & \qfsim(C) & \qfsim(Q) & Diff. & \qfsim(C) & \qfsim(Q) &  \revision{L1}    & \revision{L2}    & \revision{\qfsim(C)} & \revision{\qfsim(Q)} & \revision{Diff.} & \revision{\qfsim(C)} & \revision{\qfsim(Q)} \\
    \midrule
    \{3,6\}   & 28(4) & 12(2) & 97.10\% & 95.53\% & -1.57\% & 5.13S  & 2,555H  &  \revision{7(4)}& \revision{5(2)}& \revision{98.27\%} & \revision{97.46\%} & \revision{-0.81\%} & \revision{4.30S} & \revision{16.57H}\\
    \{3,8\}   & 28(4) & 12(2) & 86.84\% & 83.59\% & -3.25\% & 5.59S  & 2,631H  &  \revision{7(4)}& \revision{5(2)}& \revision{87.40\%} & \revision{88.06\%} & \revision{+0.54\%} & \revision{4.05S} & \revision{16.56H}\\
    \{1,3,6\} & 28(8) & 18(3) & 87.91\% & 81.99\% & -5.92\% & 15.89S & 14,650H &  \revision{7(8)}& \revision{8(3)} & \revision{88.53\%}& \revision{88.14\%} & \revision{-0.39\%} & \revision{6.96S}& \revision{47.98H} \\
    \bottomrule
    % \multicolumn{9}{l}{$^{*}:$ Qiskit Aer simulator runs 8,192 shots for each dataset}
    \end{tabular}%
  \label{tab:mnist_simulation}%
\end{table*}%

Furthermore, we have an observation for our proposed batch normalization (BN). 
\revision{For almost all test cases, BN helps to improve the accuracy of \qfnet~and \qfhnet{t}, and
the most significant improvement
is observed
at dataset $\{1,5\}$, from less than 70\% to 84.56\% for \qfnet~and 90.33\% to 96.60\% for \qfhnet{t}.} 
Interestingly, BN also helps to improve MLP(C) accuracy significantly for dataset $\{1,3,6\}$ (from less than 60\% to 81.99\%), with a slight accuracy improvement for dataset $\{3,6\}$ and a slight accuracy drop for dataset $\{3,8\}$. 
This shows that the importance of batch normalization in improving model performance and the proposed BN is definitely useful for quantum neural networks.

% , and it may also be useful even for classical neural networks, a topic worthy of future investigation.  

%All approaches without batch normalization can only achieve accuracy less than 60\%, while the accuracy can be improved to 81.99\% and 87.08\% by using batch normalization in MLP(C) and \qfnet.
%This demonstrates the importance and effectiveness of batch normalization for a more complicated model. 

\subsection*{\qfsim(C)~and \qfsim(Q) are Consistent}

Next, \revision{we evaluate the results of \qfsim(C) for both \qfnet~and \qfhnet{t} on classical computers, and that of \qfsim(Q) simulation on classical computers for the quantum circuits \qfcirc~built upon \qfsnet{s}.}
% ,  we are going to evaluate \qfsim.
Table \ref{tab:mnist_simulation} reports the comparison results in \revision{the usage of qbits in \qfcirc, inference accuracy and elapsed time}, where results under Column \qfsim(C) are the golden results.
% taking Qiskit Aer simulator as a baseline.
% comparison results between \qfsim~and the Qiskit Aer simulator on MNIST dataset.
Because of the limitation of Qiskit Aer (whose backend is ``ibmq\_qasm\_simulator'') used in \qfsim(Q) that can maximally support 32 qbits, we measure the results after each neuron.
\revision{We select three datasets, including \{3,6\}, \{3,8\}, and \{1,3,6\}, for evaluation. 
Datasets with more classes (e.g., \{0,3,6,9\}) are based on larger inputs, which will lead to the usage of qbits in \qfnet~to exceed the limitation (i.e., 32 qbits).}
Specifically, for $4\times 4$ input image in \qfnet, in the first hidden layer, it needs 23 qbits (16 input qbits, 4 encoding qbits, and 3 auxiliary qbits) for neural computation and 4 qbits for batch normalization, and 1 output qbit; as a result, it requires 28 qbits in total.
\revision{On the contrary, since \qfhnet{t} is designed in terms of the quantum circuit implementation, which takes full use of all states of $k$ qbits to represent $2^k$ data.
In consequence, the number of required qbits can be significantly reduced. 
In detail, for the $4\times 4$ input, it needs 4 qbits to represent the data, 1 output qbit, and 2 auxiliary qbits; as a result, it only needs 7 qbits in total.}
The number of qbits used for each hidden layer (``L1'' and ``L2'') is reported in column ``Qbits'', where numbers in parenthesis indicate the number of neurons in a hidden layer.

% To add a new neuron, it requires the additional 8 qbits for encoding and batch normalization, which exceeds the limitation of 32 qbit for implementing two neurons on one circuit.
% In consequence, we repeatedly simulate each neuron in \qfnet~and forward the output of neurons to the next layer. 
% , and feed them to the next layer with an additional quantum circuit.
% As we can see from this table that \qfhnet{t} requires much less number of qbits than QF-Net.
% This is because \qfhnet{t} is 

Column ``Accuracy'' in Table \ref{tab:mnist_simulation} reports the accuracy comparison. 
\revision{For \qfsim(C), there will be no difference in accuracy among different executions. 
For \qfsim(Q), we implement the obtained \qfcirc~from \qfsnet{s} on Qiskit Aer simulation with 8,192 shots.
We have the following two observations from these results: (1) There exist accuracy differences between \qfsim(C) and \qfsim(Q). This is because Qiskit Aer simulation used in \qfsim(Q) is based on the Monte Carlo method, leading to the variation.
In addition, since the output probability of different neurons may quite close in some cases, it will easily result in different classification results for small variations.
(2) Such accuracy differences for \qfhnet{t} is much less than that of \qfnet, because \qfnet~utilizes much more qbits, which leads to the accumulation of errors. 
In \qfhnet{t}, we can see that there is a small difference between \qfsim(C) and \qfsim(Q).
For the dataset \{3,8\}, \qfsim(Q) can even achieve higher accuracy.}
The above results demonstrate both \qfnet~and \qfhnet{t} can be consistently implemented on classical and quantum computers.
% As we can see from the result, the difference between \qfsim(C) and \qfsim(Q) for \qfhnet{t} is much less than that of \qfnet.
% This is because \qfnet~utilizes much more qbits, which leads the error accumulated. 
% Specifically, the results obtained by \qfsim(C) have slightly higher accuracy over \qfsim(Q).
% This is because Qiskit Aer simulation used in \qfsim(Q) is based on the Monte Carlo method, and the output probability of different neurons may quite close in some cases, leading deviations.
% Another potential issue is that the trained model is based on \qfsim; however, as will be shown later, it is not practical to employ Qiskit Aer for training due to large elapsed time.
% Nevertheless, the accuracy obtained by Qiskit Aer is consistently higher than ``MLP(C) w/ BN'' on classical computer (in Figure \ref{fig:exp_model}), achieving 0.36\%, 6.13\%, and 0.01\% accuracy gain, respectively on three datasets.
% These results again demonstrate the effectiveness of \qfnet.

\begin{figure}[t]
\centering
\includegraphics[width=1.0\linewidth]{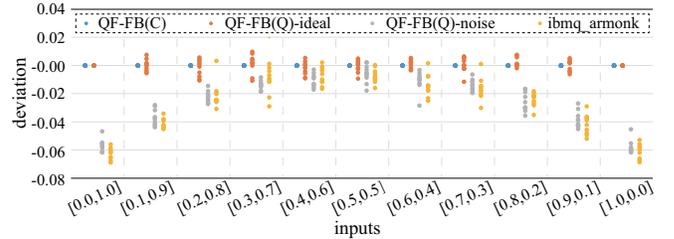}
\caption{\raggedr{Output probability comparison on \qfsim(C), \qfsim(Q)-ideal assuming perfect qbits, \qfsim(Q)-noise applying noise model for ``ibm\_armonk'' backend, and results of circuit design (``design 4'') in Figure \ref{fig:comp_verify}(d) on ``ibm\_armonk'' backend on IBM quantum processor.}}
\label{fig:exp_sim_comp}
\end{figure}

Column ``Elapsed Time'' in Table \ref{tab:mnist_simulation} demonstrates the efficiency of \qfsim.
The elapsed time is the inference time (i.e., forward propagation), used for executing all images in the test datasets, including 1968, 1983, and 3102 images for \{3,6\}, \{3,8\}, and \{1,3,6\}, respectively.
% In the experiments, \qfsim(Q) uses up 48 CPUs in the server and \qfsim(C) only uses 1 CPU.
% For fair comparison, we normalize the time to ``CPU time'' that is the product of CPU number and elapsed time.
\revision{As we can see from the table, \qfsim(Q) for \qfnet~takes over 2,500 Hours for classifying 2 digits and 14,000 Hours for classifying 3 digits, and these figures are 16 Hours and 48 Hours for \qfhnet{t}.
On the other hand, \qfsim(C) only takes less than 16 seconds for both \qfsnet{s} on all datasets.
The speedup of \qfsim(C) over \qfsim(Q) is more than six orders of magnitude larger (i.e., $10^6\times$) for \qfnet, and more than four orders of magnitude larger (i.e., $10^4\times$) for \qfhnet{t}.}
This verifies that \revision{\qfsim(C)}~can provide an efficient forward propagation procedure to support the lengthy training of \qfnet.

% An efficient simulation is quite important for  machine learning training, since the training phase is much more time-consuming (about $6\times$ more) than the inference phase.
% %images over that used in the  inference phase will be involved.
% Therefore, it is not practical to employ Qiskit Aer simulator for training, and this also emphasizes the importance of developing an efficient quantum circut simulator.

In Figure \ref{fig:exp_sim_comp}, we further verify the accuracy of \qfsim~by conducting a comparison for design 4 in Figure \ref{fig:comp_verify}(d) on IBM quantum processor with ``ibm\_armonk'' backend.
Kindly note that the quantum processor backend is selected by \qfmap.
\revision{In this experiment, the result of \qfsim(C) is taken as a baseline. 
In the figure, the x-axis and y-axis represent the inputs and deviation, respectively.
The deviation indicates the difference between the baseline and the results obtained by Qiskit Aer simulation or that by executing on IBM quantum processor.}
% ~(details see \textbf{Methods} section).
For comparison, we involve two configurations for \qfsim(Q): (1) \qfsim(Q)-ideal assuming perfect qbits; (2) \qfsim(Q)-noise with error models derived from ``ibm\_armonk''.
We launch either simulation or execution for respective approaches for 10 times, each of which is represented by a dot in Figure \ref{fig:exp_sim_comp}.
% We observe that \qfsim(C)~always generates one result while Qiskit Aer generates different results in different runs.
% This is because \qfsim~is based on theoretic formulations while Qiskit Aer is based on some Monte Carlo method.
We observe that the results of \qfsim(Q)-ideal are distributed around that generated by \qfsim(C) within 1\% deviation; while \qfsim(Q)-noise obtains similar results of that on the IBM quantum processor.
These results verify that the \qfsnet{s}~on the classical computer can achieve consistent results with that of \qfcirc~deployed on a quantum computer with perfect qbits.

\subsection*{\revision{\qfmap~is the Key to Achieve Quantum Advantage}}

\revision{Two sets of experiments are conducted to demonstrate the quantum advantage achieved by \framework.
First, we conduct an ablation study to compare the operator/gate usage of the core computation component, neural computation layer.
Then, the comparison on gate usage is further conducted on the trained neural networks for different sub-datasets from MNIST.
In these experiments, we compare \framework~to MLP(C) and FFNN(Q)\cite{tacchino2019quantum}.
For MLP(C), we consider the adder/multiplier as the basic operators, while for FFNN(Q) and \framework, we take the quantum logic gate (e.g., Pauli-X, Controlled Not, Toffoli) as the operators.
The operator usage reflects the total cycles for neural computation.
Kindly note that the results of \framework~are obtained by using \qfmap~on neural computation {U}\qfnc{R}; and that of FFNN(Q) are based on the state-of-the-art hypergraph state approach proposed in \cite{tacchino2019artificial}.
For a fair comparison, \framework~and FFNN(Q) are based on the same weights.}

\begin{figure}[t]
\centering
\includegraphics[width=1\linewidth]{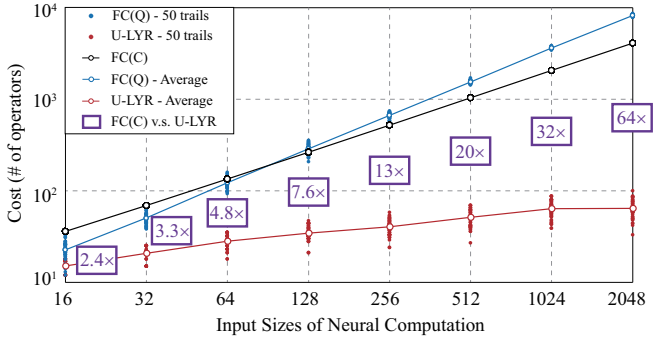}
\caption{\raggedr{\revision{Demonstration of Quantum Advantage Achieved by {U}\qfnc{R} in \framework: comparison is conducted by using 50 random generated weights for each input size.}}}
\label{fig:map_adv}
\end{figure}

% Table generated by Excel2LaTeX from sheet 'Sheet2'
\begin{table}[htbp]
  \centering
  \scriptsize
  \tabcolsep 1.2pt
  \renewcommand\arraystretch{1.3}
  \caption{\raggedr{\revision{\framework~demonstrates quantum advantages on neural networks for MNIST datasets with the increasing model sizes: comparison on the number of used gates.}}}
    \begin{tabular}{ccccccccccccccc}
    \toprule
    \multirow{2}[2]{*}{\revision{Dataset}} & \multicolumn{3}{c}{\revision{Structure}} & \multicolumn{3}{c}{\revision{MLP(C)}} & \multicolumn{4}{c}{\revision{FFNN(Q)}}  & \multicolumn{4}{c}{\revision{\textbf{\qfhnet{t}(Q)}}} \\
          & \revision{In} & \revision{L1} & \revision{L2} & \revision{L1} & \revision{L2} & \revision{Tot.} & \revision{L1} & \revision{L2} & \revision{Tot.} & \revision{\textbf{Red.}} & \revision{L1} & \revision{L2} & \revision{Tot.} & \revision{\textbf{Red.}} \\
    \midrule
    \revision{\{1,5\}} & \revision{16} & \revision{4} & \revision{2} & \multirow{4}[1]{*}{\revision{132}} & \multirow{4}[1]{*}{\revision{18}} & \multirow{4}[1]{*}{\revision{150}} & \revision{80} & \revision{38} & \revision{118} & \revision{\textbf{1.27$\times$}} & \revision{74} & \revision{38} & \revision{112} & \revision{\textbf{1.34$\times$}} \\
    \revision{\{3,6\}} & \revision{16} & \revision{4} & \revision{2} &       &       &       & \revision{96} & \revision{38} & \revision{134} & \revision{\textbf{1.12$\times$}} & \revision{58} & \revision{38} & \revision{96} & \revision{\textbf{1.56$\times$}} \\
    \revision{\{3,8\}} & \revision{16} & \revision{4} & \revision{2} &       &       &       & \revision{76} & \revision{34} & \revision{110} & \revision{\textbf{1.36$\times$}} & \revision{58} & \revision{34} & \revision{92} & \revision{\textbf{1.63$\times$}} \\
    \revision{\{3,9\}} & \revision{16} & \revision{4} & \revision{2} &       &       &       & \revision{98} & \revision{42} & \revision{140} & \revision{\textbf{1.07$\times$}} & \revision{68} & \revision{42} & \revision{110} & \revision{\textbf{1.36$\times$}} \\
    \revision{\{0,3,6\}} & \revision{16} & \revision{8} & \revision{3} & \multirow{2}[0]{*}{\revision{264}} & \multirow{2}[0]{*}{\revision{51}} & \multirow{2}[0]{*}{\revision{315}} & \revision{173} & \revision{175} & \revision{348} & \revision{\textbf{0.91$\times$}} & \revision{106} & \revision{175} & \revision{281} & \revision{\textbf{1.12$\times$}} \\
    \revision{\{1,3,6\}} & \revision{16} & \revision{8} & \revision{3} &       &       &       & \revision{209} & \revision{161} & \revision{370} & \revision{\textbf{0.85$\times$}} & \revision{139} & \revision{161} & \revision{300} & \revision{\textbf{1.05$\times$}} \\
    \revision{\{0,3,6,9\}} & \revision{64} & \revision{16} & \revision{4} & \revision{2064} & \revision{132} & \revision{2196} & \revision{1893} & \revision{572} & \revision{2465} & \revision{\textbf{0.89$\times$}} & \revision{434} & \revision{572} & \revision{1006} & \revision{\textbf{2.18$\times$}} \\
    \revision{\{0,1,3,6,9\}} & \revision{64} & \revision{16} & \revision{5} & \multirow{2}[0]{*}{\revision{2064}} & \multirow{2}[0]{*}{\revision{165}} & \multirow{2}[0]{*}{\revision{2229}} & \revision{1809} & \revision{645} & \revision{2454} & \revision{\textbf{0.91$\times$}} & \revision{437} & \revision{645} & \revision{1082} & \revision{\textbf{2.06$\times$}} \\
    \revision{\{0,1,2,3,4\}} & \revision{64} & \revision{16} & \revision{5} &       &       &       & \revision{1677} & \revision{669} & \revision{2346} & \revision{\textbf{0.95$\times$}} & \revision{445} & \revision{669} & \revision{1114} & \revision{\textbf{2.00$\times$}} \\
    \revision{\{0,1,3,6,9\}$^*$} & \revision{256} & \revision{8} & \revision{5} & \revision{4104} & \revision{85} & \revision{4189} & \revision{5030} & \revision{251} & \revision{5281} & \revision{\textbf{0.79$\times$}} & \revision{135} & \revision{251} & \revision{386} & \revision{\textbf{10.85$\times$}} \\
    \midrule
    \multicolumn{15}{l}{\revision{$^*$: Model with $16\times 16$ resolution input for dataset \{0,1,3,6,9\} to test scalability, whose}} \\
    \multicolumn{15}{l}{\revision{accuracy is 94.09\%, which is higher than $8\times 8$ input with accuracy of 92.62\%.}} \\
    \bottomrule
    \end{tabular}%
  \label{tab:qfmap}%
\end{table}%

\revision{Figure \ref{fig:map_adv} reports the comparison results for the core component in neural network, the neural computation layer.
The x-axis represents the input size of the neural computation, and the y-axis stands for the cost, that is, the number of operators used in the corresponding design.
For quantum implementation (both FC(Q)\cite{tacchino2019artificial} in FFNN(Q)\cite{tacchino2019quantum} and {U}\qfnc{R} in \framework), the value of weights will affect the gate usage, so we generate 50 sets of weights for each scale of input, and the dots on the lines in this figure represent average cost. 
From this figure, it clearly shows that the cost of FC(C) in MLP(C) on classical computing platforms grows exponentially along with the increase of inputs.
The state-of-the-art quantum implementation FC(Q) has the similar exponentially growing trend.
On the other hand, we can see that the growing trend of {U}\qfnc{R} is much slower.
As a result, the cost reduction continuously increases along with the growth of the input size of neural computation.
For the input size of 16 and 32, the average cost reductions are 2.4$\times$ and 3.3$\times$, compared with the implementations on classical computers.
When the input size grows to 2,048, the cost reduction increased to 64$\times$ on average.
The cost reduction trends in this figure clearly demonstrate the quantum advantage achieved by {U}\qfnc{R}.
In the \textbf{Methods} section, for the neural computation with an input size of $2^k$, we will show that the complexity for quantum implementation is $O(k^2)$, while it is $O(2^k)$ for classical computers.}

% We obtain the consistent results with that in Table \ref{tab:qfmap}, where the average cost reductions are 2.4$\times$ and 4.8$\times$ for input size of 16 and 64, compared with the implementations on classical computers.

% further evaluates the scale of \framework~for neural computation with larger inputs.
% In this figure, x-axis and y-axis represent the input size and the number of operators, respectively. 
% We obtain the consistent results with that in Table \ref{tab:qfmap}, where the average cost reductions are 2.4$\times$ and 4.8$\times$ for input size of 16 and 64, compared with the implementations on classical computers.
% The cost reduction sharply increases with the input size.
% When the input size grows to 2,048, the cost reduction increased to 64$\times$ on average.
% The cost reduction trends in this figure clearly demonstrate the quantum advantage achieved by \framework.

\revision{Table \ref{tab:qfmap} reports the comparison results for the whole network. 
The neural network models for MNIST in Figure \ref{fig:exp_model} are deployed to quantum circuits to get the cost.
In addition, to demonstrate the scalability, we further include a new model for dataset ``\{0,1,3,6,9\}$^{*}$'', which takes the larger sized inputs but less neurons in the first layer $L1$ and having higher accuracy over ``\{0,1,3,6,9\}''.
In this table, columns $L1$, $L2$, and $Tot.$ under three approaches report the number of gates used in the first and second layers, and in the whole network.
Columns ``Red.'' represent the comparison with baseline $MLP(C)$.}

\revision{From the table, it is clear to see that all cases implemented by \qfhnet{t} can achieve cost reduction over MLP(C), while for datasets with more than 3 classes, FFNN(Q) needs more gates than MLP(C).
A further observation made in the results is that \qfhnet{t} can achieve higher cost reduction with the increase of input size.
Specifically, for input size is 16, the reduction ranges from $1.05\times$ to $1.63\times$. The reduction increases to $2.18\times$ for input size is 64, and it continuously increases to $10.85\times$ when the input size grows to 256.
The above results are consistent with the results shown in Figure \ref{fig:map_adv}.
It further indicates that even the second layer in \qfhnet{t} uses the {P}\qfnc{R} which requires more gates for implementation, the quantum advantage can still be achieved for the whole network because the first layer using {U}\qfnc{R} can significantly reduce the number of gates.
Above all, \framework~demonstrates the quantum advantages on MNIST dataset. 
}

% In this table, columns ``In.'' and ``Out.'' represent the number of input neurons and output neurons in the considering layer.
% We further divide the neural computation to a linear function (columns ``Lin.'') and non-linear function (columns ``Squ.'').
% Columns ``Tot.'' reports the total number of operators used in different platforms and designs.
% From the column ``Red.'' in the table, it is clear to see that \framework~can achieve significant reductions in the number of operators used for neural computation. 
% Compared with MLP(C), for the number of input neurons to be 16, \framework~achieves roughly 2$\times$ reduction.
% When the number of input neurons increases to 64, the reduction is increased to over $4\times$.
% Another observation made in the results is that the operator reduction of FFNN(Q) decreases from over 1.5$\times$ for input size of 16 to 1.09$\times$ for input size of 64.
% Above all, \qfmap~demonstrates the quantum advantages in neural computation. 
% In the \textbf{Methods} section, for the neural computation with an input size of $2^k$, we will show that the complexity for quantum implementation is $O(k^2)$, while it is $O(2^k)$ for classical computers.
% }

\subsection*{\qfcirc~on IBM Quantum Processor}

This subsection further evaluates the efficacy of \framework~on IBM Quantum Processors.
We first show the importance of quantum circuit optimization in \qfcirc~ to minimize the number of required qbits.
% and verify the accuracy of \qfsim~in simulation.
Based on the optimized circuit design, we then deploy a 2-input binary classifier on IBM quantum processors.

Figure \ref{fig:comp_verify} demonstrates the optimization of a 2-input neuron step by step.
All quantum circuits in Figures \ref{fig:comp_verify}(a)-(d) achieve the same functionality, but with a different number of required qbits.
The equivalency of all designs will be demonstrated in the \textbf{Supplementary Information}.
Design 1 in Figure \ref{fig:comp_verify}(a) is directly derived from the design methodology presented in  \textbf{Methods} section.
To optimize the circuit using fewer qbits, we first convert it to the circuit in Figure \ref{fig:comp_verify}(b), denoted as design 2.
Since there is only one controlled-Z gate from qbit I0 to qbit E/O, we can merge these two qbits, and obtain an optimized design in Figure \ref{fig:comp_verify}(c) with 2 qbits, denoted as design 3.
The circuit can be further optimized to use 1 qbit, as shown in Figure \ref{fig:comp_verify}(d), denoted as design 4.
The function $f$ in design 4 is defined as follows:
\begin{equation}\label{equ:funcf}
f(\alpha,\beta)=2\cdot arcsin(\sqrt{x+y-2\cdot x\cdot y}),
\end{equation}
% \begin{equation}
% f(\alpha,\beta)=2\cdot arcsin(\sqrt{sin^2\frac{\alpha}{2}+sin^2\frac{\beta}{2}-2\cdot sin^2\frac{\alpha}{2}\cdot sin^2\frac{\beta}{2}})
% \end{equation}
where $x=sin^2\frac{\alpha}{2}$, $y=sin^2\frac{\beta}{2}$, representing input probabilities.

To compare these designs, we deploy them onto IBM Quantum Processors, where ``ibm\_velencia'' backend is selected by \qfmap.
In the experiments, we use the results from \qfsim(C)~as the golden results.
\revision{Figure \ref{fig:comp_verify}(e) reports the deviations of design 1 and design 4 against the golden results.
The results clearly show that design 4 is more robust because it uses fewer qbits in the circuit.}
Specifically, the deviation of design 4 against golden results is always less than 5\%, while reaching up to 13\% for design 1.
In the following experiments, design 4 is applied in \qfcirc.

\begin{figure}[t]
\centering
\includegraphics[width=1\linewidth]{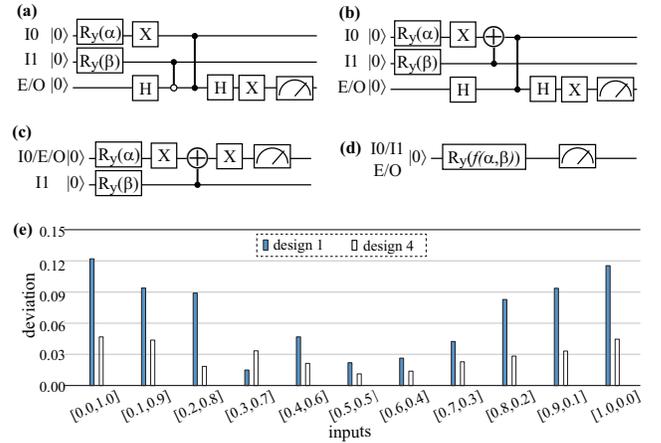}
\caption{\raggedr{Evaluation of the quantum circuits for a two-input neural computation, where weights are \{-1,+1\}: (a) design 1: original neural computation design; (b-d) three optimized designs (design 2-4), based on design 1; \revision{(e) the deviation of design 1 and design 4 obtained from ``ibm\_velencia'' backend IBM quantum processor, using \qfsim(C)~as golden results.}}}
% ; (f) output probability comparison on \framework~simulation, Qiskit ideal simulation, Qiskit simulation with noise, and results of design 4 on ``ibm\_armonk'' backend.}}
% output probabilities of design 1-4 obtained from ``ibm\_velencia'' backend IBM quantum processor, Qiskit simulator, and our simulator, where the number of shots are set as 8,192; (f) the difference of output probability, taking our simulator as a baseline;}}
\label{fig:comp_verify}
\end{figure}

% \clearpage
\begin{figure*}[t]
\centering
\includegraphics[width=0.95\linewidth]{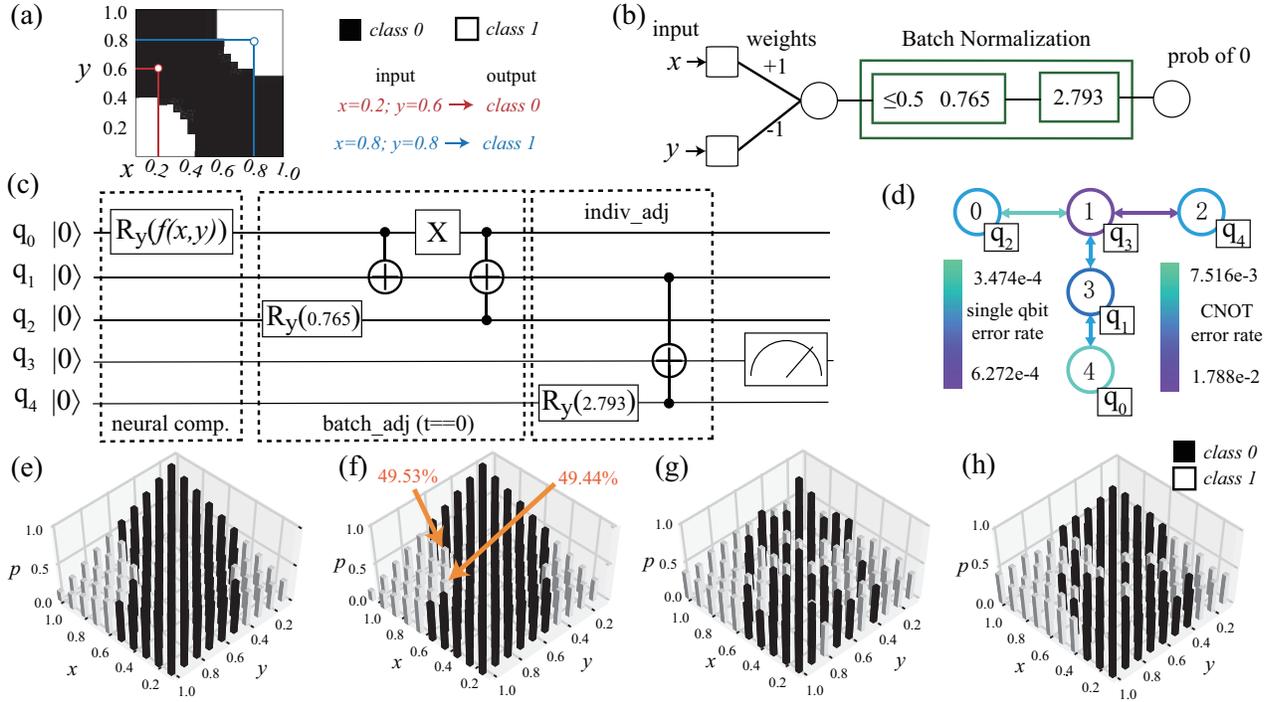}
\caption{\raggedr{Results of a binary classification case study on IBM quantum processor of ``ibmq\_essex'' backend: (a) binary classification with two inputs ``$x$'' and ``$y$''; (b) \qfsnet{s}~with trained parameters; (c) \qfcirc~derived from the trained \qfsnet{s}; (d) the virtual-to-physic mapping obtained by \qfmap~upon ``ibmq\_essex'' quantum processor; (e) \qfsim(C) achieves 100\% accuracy; (f) \qfsim(Q) achieves 98\% accuracy where 2 marked error cases having probability deviation within 0.6\% ; (g) results on ``ibmq\_essex'' using the default mapping, achieving 68\% accuracy; (h) results obtained by ``ibmq\_essex'' with the mapping in (d), achieving 82\% accuracy; shots number in all tests is set as 8,192.}}
% QNN achieves the state-of-the-art accuracy in image classifications on MNIST dataset: (a) comparison with Multi-Level Perceptron (MLP) on classical computer; (b) comparison with feed-forward neural network (FFNN) \cite{xxx} designed for quantum computer.}}
\label{fig:exp_class}
\end{figure*}

Next, we are ready to introduce the case study on an end-to-end binary classification problem as shown in Figure \ref{fig:exp_class}.
In this case study, we train the \qfnet~based on \qfsim(C).
Then, the tuned parameters are applied to generate \qfcirc.
Finally, \qfmap~optimizes the deployment of \qfcirc~to IBM quantum processor, selecting the ``ibmq\_essex'' backend.

% reports the results of \framework~in solving a binary classification problem, which is implemented on IBM Quantum Experience using ``ibmq\_essex'' backend.
The classification problem is illustrated in Figure \ref{fig:exp_class}(a), which is a binary classification problem (two classes) with two inputs: $x$ and $y$. 
For instance, if $x=0.2$ and $y=0.6$, it indicates class 0.
The \qfnet, \qfcirc, and \qfmap~are demonstrated in Figure \ref{fig:exp_class}(b)-(d).
First, Figure \ref{fig:exp_class}(b) shows that \qfnet~consists of one hidden layer with one 2-input neuron and batch normalization.
% (detail designs see \textbf{Methods} section).
The output is the probability $p_0$ of class 0. Specifically, an input is recognized as class 0 if $p_0\ge 0.5$; otherwise it is identified as class 1.

The quantum circuit \qfcirc~of the above \qfnet~is shown in Figure \ref{fig:exp_class}(c).
The circuit is composed of three parts, (1) neural computation, (2) batch\_adj in batch normalization, and (3) indiv\_adj in batch normalization.
The neural computation is based on design 4  as shown in Figure \ref{fig:comp_verify}(d).
The parameter of $Ry$ gate in neural computation at qbit $q_0$ is determined by the inputs $x$ and $y$.
Specifically, $f(x,y) = 2\cdot arcsin(\sqrt{x+y-2\cdot x\cdot y})$, as shown in Formula \ref{equ:funcf}.
Then, batch normalization is implemented in two steps, where qbits $q_2$ and $q_4$ are initialized according to the trained BN parameters.
During the process, $q_1$ holds the intermediate results after batch\_adj, and $q_3$ holds the final results after indiv\_adj.
Finally, we measure the output on qbit $q_3$\footnote{A Quirk-based example of inputs 0.2 and 0.6 leading to $f(x,y)=1.6910$ can be accessed by \url{https://wjiang.nd.edu/quirk_0_2_0_6.html}, which is accessible at 06-19-2020. The output probability of 60.3\% is larger than 50\%, implying the inputs belong to class 0.}.

% (i.e., $p_2=0.1393$ and $p_4=0.9699$), which are fixed and computed by $arccos(1-p_2\times 2)$ and $arccos(p_4\times 2-1)$. 
% The slight difference is caused by the optimization of circuits, which will be discussed in \textbf{Methods} section.
% The encoder is optimized to be implemented on the input qbit and the $Z$ gate between two $H$ gates represents the weight $+1$ and $-1$ on two inputs.

After building \qfcirc, the next step is to map qbits from the designed circuit to the physic qbits on the quantum processor, and this is achieved through our \qfmap.
In this experiment, \qfmap~selects ``ibm\_essex'' as backend with its physical properties shown in Figure \ref{fig:exp_class}(d), where error rates of each qbit and each connection are illustrated by different colors.
By following the rules as defined by \qfmap~(see Method section), we obtain the physically mapped \qfcirc~ shown in Figure \ref{fig:exp_class}(d). For example,  the input $q_0$ is mapped to the physical qbit labeled as 4.

% \qfmap~has two tasks: (1) select a suitable backend; (2) map virtual qbits to phsyic qbits in the selected backend.
% For the first task, \qfmap~will \textit{i)} check the number of qbits needed; \textit{ii)} find the backend with smallest number of qbit to accommodate \qfcirc; \textit{iii)} for the backends with the same number of qbits, \qfmap~will select backend for the minimum average error rate.
% Here, we need to use 5 qbits, and ``ibm\_essex'' backend is selected.
% Figure \ref{fig:exp_class}(d) demonstrates the ``ibm\_essex'' quantum processor.
% % These qbits are connected as a ``T'' shape.
% The error rates of each qbit and each connection are illustrated by different colors.
% The second task in \qfmap~will map virtual qbits to physic qbits. 
% The mapping follows two rules: (1) the virtual qbit with more gates is mapped to the physic qbit with a lower error rate, say $q_0\rightarrow 4$; (2) virtual qbits with connections in the circuit are mapped to the physic qbits with the smallest distance, say $\langle q_0,q_1\rangle \rightarrow \langle 4,3\rangle$.

After \framework~goes through all the steps from input data to the physic quantum processor, we can perform inference on the quantum computer. 
In this experiments, we test 100 combinations of inputs from $\langle x,y\rangle = \langle 0.1,0.1\rangle$ to $\langle x,y\rangle = \langle 1.0,1.0\rangle$.
First, we obtain the results using \qfsim(C) as golden results~and \qfsim(Q) as quantum simulation assuming perfect qbits, which are reported in Figure \ref{fig:exp_class}(e) and (f), achieving 100\% and 98\% prediction accuracy.
The results verify the correctness of the proposed \qfnet.
Second, the results obtained on quantum processors are shown in Figure \ref{fig:exp_class}(h), which achieves 82\% accuracy in prediction.
For comparison, in Figure \ref{fig:exp_class}(g), we also show the results obtained by using the default mapping algorithm in IBM Qiskit, whose accuracy is only 68\%.
This result demonstrates the value of \qfmap~in further improving the physically achievable accuracy on a physical quantum processor with errors.

\section*{Discussion}
In summary, we propose a holistic  \framework~framework to co-design the neural networks and quantum circuits.
\revision{Novel quantum-aware \qfsnet{s}~are first designed}. 
Then, an accurate and efficient inference engine, \qfsim, is proposed to enable the training of \qfsnet{s}~on classical computers.
\revision{Based on \qfsnet{s}~and the training results, the \qfmap~can automatically generate and optimize a corresponding quantum circuit, \qfcirc}.
Finally, \qfmap~can further map \qfcirc~to a quantum processor in terms of qbits' error rates.

\begin{table}[t]
  \centering
  \scriptsize
  \tabcolsep 3pt
  \renewcommand\arraystretch{1.3}
  \caption{\raggedr{\revision{Comparison of the implementation of Neural Computation with $m=2^k$ input neurons.}}}
    \begin{tabular}{cccccc}
    \toprule
    \multicolumn{2}{c}{\revision{Layers}} & \revision{FC(C)\cite{rosenblatt1957perceptron}} & \revision{FC(Q)\cite{tacchino2019quantum}} & \revision{{P}\qfnc{R}} & \revision{{U}\qfnc{R}} \\
    \midrule
    \multirow{2}[1]{*}{\revision{Complexity}} & \revision{\# Bits/Qbits} & \revision{$O(2^k)$} & \revision{$O(k)$} & \revision{$O(2^k)$} & \revision{\textbf{$O(k)$}} \\
          & \revision{\# Operators} & \revision{$O(2^k)$} & \revision{$O(2^k)$} & \revision{$O(k\cdot 2^k)$} & \revision{\textbf{$O(k^2)$}} \\
    \multirow{2}[0]{*}{\revision{Data Representation}} & \revision{Input Data} & \revision{F32} & \revision{Bin} & \revision{R.V.} & \revision{F32} \\
          & \revision{Weights} & \revision{Bin (F32)} & \revision{Bin} & \revision{\textbf{Bin (R.V.)}} & \revision{Bin} \\
    \multicolumn{2}{c}{\revision{Connect Layers w/o Measurement}}  & \revision{\checkmark} & \revision{-} & \revision{\checkmark} & \revision{$\times$} \\
    % \multirow{2}[0]{*}{\revision{Non-Linear Function}} & \revision{Quadratic} & \revision{\checkmark} & \revision{\checkmark} & \revision{\checkmark} & \revision{\checkmark} \\
    %       & \revision{ReLU} & \revision{\checkmark} & \revision{$\times$} & \revision{\textbf{\checkmark}} & \revision{$\times$} \\
    \midrule
    \multirow{2}[1]{*}{\revision{Summary}} & \revision{Flexibility} & \revision{-} & \revision{$\times$} & \revision{\textbf{\checkmark}} & \revision{$\times$} \\
          & \revision{Qu. Adv.} & \revision{-} & \revision{$\times$} & \revision{$\times$} & \revision{\textbf{\checkmark}} \\
    \bottomrule
    \end{tabular}%
  \label{tab:nccomp}%
\end{table}%

The neural computation layer is one key component in \framework~to achieve state-of-the-art accuracy and quantum advantage. 
We have shown in Figure \ref{fig:exp_model} that the existing quantum-aware neural network \cite{tacchino2019quantum} that interprets inputs as the binary form will degrade the network accuracy.
\revision{To address this problem, in \qfnet, we first propose a probability-based neural computation layer, denoted as {P}\qfnc{R}, which interprets real number inputs as random variables following a two-point distribution.
As shown in Table \ref{tab:nccomp}, {P}\qfnc{R} can represent both input and weight data using random variables, and it can 
directly connect layers without measurement.
% support both Quadratic and ReLU non-linear functions.
% Kindly note that since Quadratic function is more nature on quantum implementation and works for other designs, in this work, we use Quadratic function rather than ReLU.
In summary, {P}\qfnc{R} provides better flexibility to perform neural computation than others; however, it suffers high complexity, i.e., $O(2^k)$ for the usage of qbits and $O(k\times2^k)$ for the usage of operators (basic quantum gates).}

\revision{In order to acquire quantum advantages, we further propose a unitary matrix based neural computation layer, called {U}\qfnc{R}.
As illustrated in Table \ref{tab:nccomp}, {U}\qfnc{R} sacrifices some degree of flexibility on data representation and non-linear function but can significantly reduce the circuit complexity.
Specifically, with the help of \qfmap, the number of basic operators used by {U}\qfnc{R} can be reduced from $O(2^k)$ to $O(k^2)$, compared to FC(C) and FC(Q).
Kindly note that this work does not take the cost of inputs encoding into consideration in demonstrating quantum advantage; instead, we focus on the speedup of the commonly used computation component layer, that is, the neural computation layer.
The cost of encoding inputs can be reduced to $O(1)$ by preprocessing data and storing them into quantum memory, or approximating the quantum states by using basic gate (e.g., Ry).
For neural computation, we demonstrated that {U}\qfnc{R} can successfully achieve quantum advantage in the next section.}

% % It shows the potential quantum advantage achieved by {U}\qfnc{R}.
% \revision{We draw a conclusion that for neural computation with large input size, we prefer {U}\qfnc{R}, while for that with small input size but need more flexibility, we prefer to use {P}\qfnc{R}.
% \qfhnet{t} takes advantage of both neural computation paradigms and achieves the highest accuracy, as shown in Figure \ref{fig:exp_model}.
% Details of {P}\qfnc{R}, {U}\qfnc{R}, and \qfmap~will be introduced in the next section.}

\revision{Batch normalization is another key technique in improving accuracy, since the backbone of the quantum-friendly neuron computation layers ({P}\qfnc{R} and {U}\qfnc{R}) is similar to that in classical computers, using both linear and non-linear functions.}
% , particularly when the network grows deeper.
This can be seen from the results in Figure \ref{fig:exp_model}.
\revision{Batch normalization can achieve better model accuracy, mainly because the data passing a nonlinear function $y^2$ will lead to outputs to be significantly shrunken to a small range around 0 for real number representation and $1/m$ for a two-point distribution representation, where $m$ is the number of inputs.}
%this can be seen from Figure \ref{fig:stream}(d), where the term $m$ dominates the results.
% , where the term $m$ (i.e., $N$ in this case) dominates the results.
Unlike straightforwardly doing normalization on classical computers, it is non-trivial to normalize a set of qbits.
Innovations are made in \framework~for a quantum-friendly normalization.

\revision{The philosophy of co-design is demonstrated in the design of {P}\qfnc{R}, {U}\qfnc{R}, and \qfbn.
From the neural network design, we take the known operations as the backbones in {P}\qfnc{R}, {U}\qfnc{R}, and \qfbn; while from the quantum circuit design, we take full use of its ability in processing probabilistic computation and unitary matrix based computations to make {P}\qfnc{R}, {U}\qfnc{R}, and \qfbn~quantum-friendly.
In addition, as will be shown in the next section, the key to achieve quantum advantage for {U}\qfnc{R} is that \qfmap~fully considers the flexibility of the neural networks (i.e., the order of inputs can be changed), while the requirement of continuously executing machine learning algorithms on the quantum computer leads to a hybrid neural network, \qfhnet{t}, with both neural computation operations: {P}\qfnc{R} and {U}\qfnc{R}.
% and quantum computing can provide the highly parallelism in computing all states by operating unitary matrix.
Without the co-design, the previous works did not exploit quantum advantages in implementing neural networks on quantum computers, which reflects the importance of conducting co-design.} 

We have experimentally tested \framework~on a 32-qbit Qiskit Aer simulator and a 5-qbit IBM quantum processor based on superconducting technology.
We show that the proposed quantum oriented neural networks \qfsnet{s} can obtain state-of-the-art accuracy on the MNIST dataset.
It can even outperform the conventional model on a similar scale for the classical computer.
For the experiments on IBM quantum processors, we demonstrate that, even with the high error rates of the current quantum processor, \qfsnet{s}~can be applied to classification tasks with high accuracy.

In order to accelerate the \qfsim~on classical computers to support training, we 
% an efficient quantum simulator to support the machine learning training procedure, however, \qfsim~has made the following two assumptions.
% First, we assume that the outputs from a hidden neuron layer (also the inputs to the next neuron layer) are independent of each other. 
% This simplifies the simulation procedure but leads to the size of the quantum circuit without measurement to increase exponentially in terms of the number of layers.
% However, the independence of neural outputs provides another design option that can separate the computation of different layers by measuring the output of each neural computation.
% This can significantly reduce the number of required qbits and take effects on the near-term quantum computers.
% This requires much fewer number of qbits, and for the near quantum computer with limited number of qbits.
% Second, we 
make the assumptions that the perfect qbits are used. 
This enables us to apply theoretic formulations to accelerate the simulation process; however, it leads to some error in predicting
the outputs of its corresponding deployment on a physical quantum processor with high error rates
(such as the current IBM quantum processor with error rates in the
range of $10^{-2}$).
However, we do not deem
this as a drawback of our approach, rather this is an inherent problem of the current physical implementation of quantum processors. As the error rates get smaller in the future, it will help to narrow the gap between what \qfsnet{s}~predicts and what quantum processor delivers.
%ons when \qfnet~is trained based on the perfect qbits and deployed to the current IBM quantum processor with high error rate (approaching to $10^-2$).
%However, as we can see from our case study results, the \qfnet~trained based on the perfect qbits can achieve 82\% accuracy.
With the innovations on reducing the error rate of physic qbits, \qfsnet{s} will achieve better results.

\begin{figure}[t]
\centering
\includegraphics[width=1.0\linewidth]{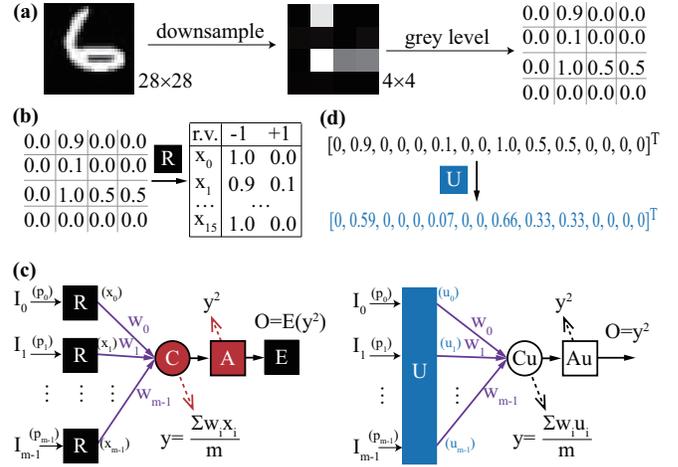}
\caption{\raggedr{\revision{Neural Computation: (a) prepossessing of inputs by \textit{i)} down-sampling the original $28\times 28$ image in MNIST to $4\times 4$ image and \textit{ii)} get the $4\times 4$ matrix with grey level normalized to $[0,1]$.
(b-c) {P}\qfnc{R}: (b) input data are converted from real number to a random variable following a two-point distribution; (c) four operations in {P}\qfnc{R}, \textit{i)} \textit{R:} converting a real number ranging from 0 to 1 to a random variable, \textit{ii)} \textit{C:} average sum of weighted inputs, \textit{iii)} \textit{A:} non-linear activation function, \textit{iv)} \textit{E:} converting random variable to a real number.
(d-e) {U}\qfnc{R}: (d) $m$ input data are converted to a vector in the first column of a $m\times m$ unitary matrix; (e) three operations in {U}\qfnc{R}, \textit{i)} \textit{U:} unitary matrix converter, \textit{ii)} \textit{$C_u$:} average sum of weighted inputs; \textit{iii)} \textit{$A_u$:} non-linear activation function. 
}}}
% (c) batch normalization after neural computation with two sub-components: \textit{i)} batch adjustment batch\_adj and \textit{ii)} individual adjustment indiv\_adj.}}}
% ; (d) a multi-layer neural network with 2 hidden layers, composed of fundamental components.}}
\label{fig:method_qf_net}
\end{figure}

% \clearpage

\section*{Methods}

We are going to introduce \framework~in this section.
Neural computation and batch normalization are two key components in a neural network, and we will present the design and implementation of these two components in \qfsnet{s}, \qfsim, \qfcirc, and \qfmap, respectively.

\subsection*{\revision{\qfnet~and \qfhnet{t}}}

\revision{Figure \ref{fig:method_qf_net} demonstrates two different neural computation components in \framework: {P}\qfnc{R} and {U}\qfnc{R}.
As stated in the \textbf{Discussion} section, {P}\qfnc{R} and {U}\qfnc{R} have their different features.
Before introducing these two components, we demonstrate the common prepossessing step in Figure \ref{fig:method_qf_net}(a), which goes through the downsampling and grey level normalization to obtain a matrix with values in the range of 0 to 1.
With the prepossessed data, we will discuss the details of each component in the following texts.}

% two fundamental components in the proposed quantum-friendly neural network \qfnet: (1) neural computation in Figure \ref{fig:method_qf_net}(a); and batch normalization in Figure \ref{fig:method_qf_net}(c).
% % Neural computation is conducted in 4 steps; while batch normalization is composed of 2 sub-components.
% % Based on the design of fundamental components, a multi-layer network can be constructed, as shown in Figure \ref{fig:method_qf_net}(d).
% We discuss the details of each component as follows. 

\begin{figure}[t]
\centering
\includegraphics[width=0.7778\linewidth]{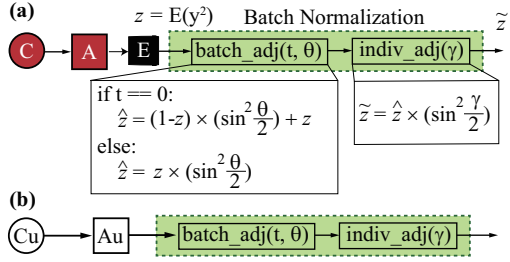}
\caption{\raggedr{\revision{Quantum implementation aware batch normalization: (a) connected to {P}\qfnc{R}; (b) connected to {U}\qfnc{R}.
}}}
\label{fig:method_qf_bn}
\end{figure}

\revision{\textbf{Neural Computation {P}\qfnc{R}:}} An m-input neural computation component is illustrated in \ref{fig:method_qf_net}(d), where $m$ input data $I_0,I_1,\cdots,I_{m-1}$ and $m$ corresponding weights $w_0,w_1,\cdots,w_{m-1}$ are given.
Input data $I_i$ is a real number ranging from 0 to 1, while weight $w_i$ is a $\{-1,+1\}$ binary number.
Neural computation in {P}\qfnc{R}~is composed of 4 operations: \textit{i)} \underline{\textbf{R}}: this operation converts a real number $p_k$ of input $I_k$ to a two-point distributed random variable $x_k$, where $P\{x_k=-1\}=p_k$ and $P\{x_k=+1\}=1-p_k$,
% model input data $I_k$ as a random variable $x_k$ following a two-point distribution with outcomes of $-1$ and $+1$, 
as shown in \ref{fig:method_qf_net}(b).
For example, we treat the input $I_0$'s
real value of $p_0$ as the probability of $x_0$
that outcomes $-1$ while $q_0=1-p_0$ as the probability that outcomes $+1$.
\textit{ii)} \underline{\textbf{C}}: this operation calculates $y$ as the average sum of weighted inputs, where the weighted input is the product of a converted input (say $x_k$) and its corresponding weight (i.e., $w_k$).
Since $x_k$ is a two-point random variable, whose values are $-1$ and $+1$ and the weights are binary values of $-1$ and $+1$, if $w_k=-1$, $w_k\cdot x_k$ will lead to the swap of probabilities $P\{x_k=-1\}$ and $P\{x_k=+1\}$ in $x_k$.
\textit{iii)} \underline{\textbf{A}}: we consider the quadratic function as the non-linear activation function in this work, and $A$ operation outputs $y^2$ where $y$ is a random variable.
\textit{iv)} \underline{\textbf{E}}: this operation converts the random variable $y^2$ to 0-1 real number by taking its expectation.
It will be passed to batch normalization to be further used as the input to the next layer.

\begin{figure*}[t]
\centering
    \includegraphics[width=1.0\linewidth]{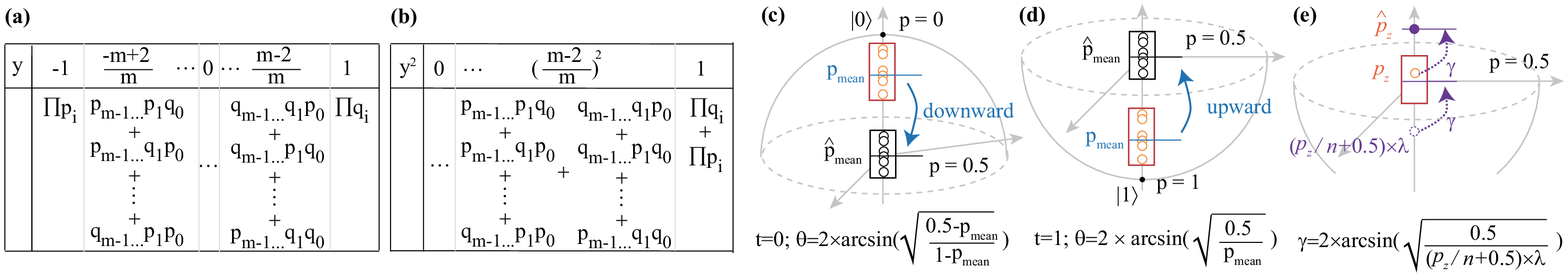}
\caption{\raggedr{\qfsim: (a-b) distribution of random variable $y$ and $y^2$ in neural computation component of \qfnet; (c-e) determination of parameters, $t$, $\theta$, and $\gamma$, in batch normalization component of \qfsnet{s}.}}
\label{fig:method_qf_sim}
\end{figure*}

\revision{\textbf{Neural Computation {U}\qfnc{R}}: Unlike {P}\qfnc{R} taking advantage of the probabilistic 
properties of qbits to provide the maximum flexibility, {U}\qfnc{R} aims to minimize the gates for quantum advantage using the property of the unitary matrix.
The $2^k$ input data are first converted to $2^k$ corresponding data that can be the first column of a unitary matrix, as shown in Figure \ref{fig:method_qf_net}(d).
Then the linear function $C_u$ and activation quadratic function $A_u$ are conducted.
{U}\qfnc{R} has the potential to significantly reduce the quantum gates for computation, since the $2^k$ inputs are the first column in a unitary matrix and can be encoded to $k$ qbits.
But the state-of-the-art hypergraph based approach\cite{tacchino2019artificial} needs $O(2^k)$ basic quantum gates to encode $2^k$ corresponding weights to $k$ qbits, which is the same with that of classical computer needing $O(2^k)$ operators (i.e., adder/multiplier).
In the later section of \qfmap, we propose an algorithm to guarantee that the number of used basic quantum gates to be $O(k^2)$, achieving quantum advantages.}

\revision{
\textbf{Multiple Layers:} 
{P}\qfnc{R} and {U}\qfnc{R} are the fundamental components in \qfsnet{s}, which may have multiple layers.
In terms of how a network is composed using these two components, we present two kinds of neural networks: \qfnet~and \qfhnet{t}.
\qfnet~is composed of multiple layers of {P}\qfnc{R}.
% Kindly note that, for multiple layers of {P}\qfnc{R} on classical computer, \underline{\textbf{R}}, \underline{\textbf{C}}, \underline{\textbf{A}}, and \underline{\textbf{E}} operations will be sequentially executed in each layer. 
% However, f
For its quantum implementation, operations on random variables can be directly operated on qbits.
Therefore, \underline{\textbf{R}} operation is only conducted in the first layer. 
Then, \underline{\textbf{C}} and \underline{\textbf{A}} operations will be repeated without measurement.
Finally, at the last layer, we measure the probability for output qbits, which is corresponding to the \underline{\textbf{E}} operation.
On the other hand, \qfhnet{t} is composed of both {U}\qfnc{R} and {P}\qfnc{R}, where the first layer applies {U}\qfnc{R} with the converted inputs. 
The output of {U}\qfnc{R} is directly represented by the probability form on a qbit, and it can seamlessly connect to \underline{\textbf{C}} in {P}\qfnc{R} used in later layers.}

% \begin{figure}[t]
% \centering
% \includegraphics[width=1.0\linewidth]{Figure/QF_Net.eps}
% \caption{\raggedr{\qfnet: (a) neural computation with four operations, \textit{i)} convert from 0-1 real number to random variable, \textit{ii)} average sum of weighted inputs, \textit{iii)} non-linear activation function; \textit{iv)} convert from random variable to 0-1 real number;
% (b) data representation of random variable followed two-point distribution; (c) batch normalization with two sub-components: \textit{i)} batch adjustment, batch\_adj, and \textit{ii)} individual adjustment, indiv\_adj; (d) a multi-layer neural network with 2 hidden layers, composed of fundamental components.}}
% \label{fig:method_qf_net}
% \end{figure}

\textbf{Batch Normalization}: \revision{Figure \ref{fig:method_qf_bn} illustrates the proposed batch normalization (\qfbn) component.
It can take the output of either {P}\qfnc{R} or {U}\qfnc{R} as input.}
\qfbn~is composed of two sub-components: batch adjustment (``batch\_adj'') and individual adjustment (``indiv\_adj'').
Basically, batch\_adj is proposed to 
% aims to normalize the probability mean of a batch of outputs to 0.5 to 
avoid data to be continuously shrunken to a small range (as stated in \textbf{Discussion} section).
This is achieved by normalizing the probability mean of a batch of outputs to 0.5 at the training phase, as shown in Figure \ref{fig:method_qf_sim}(c)-(d).
In the inference phase, the output $\hat{z}$ can be computed as follows: 
\begin{equation}\label{equ:batch_adj}
\begin{aligned}
\hat{z}&=(1-z)\times(sin^2\frac{\theta}{2})+z,\ \ \ \ if\ t=0 \\  \hat{z}&=z \times(sin^2\frac{\theta}{2}),\ \ \ \ \ \ \ \ \ \ \ \ \ \ \ \ \ \ \ if\ t=1
\end{aligned}
\end{equation}

After batch\_adj, the outputs of all neurons are normalized around 0.5.
In order to increase the variety of different neurons' output for better classification, indiv\_adj is proposed.
It contains a trainable parameter $\lambda$ and a parameter $\gamma$ (see Figure \ref{fig:method_qf_sim}(e)).
\revision{It is performed in two steps: (1) we get a start point of an output $p_z$ according to $\lambda$, and then moves it back to p=0.5 to obtain parameter $\gamma$; (2) we move $p_z$ the angle of $\gamma$ to obtain the final output. 
Since different neurons have different values of $\lambda$, the variation of outputs can be obtained.}
% while 
% % adjust a batch of outputs from a neuron, while
% indiv\_adj is proposed to increase the variety of different neurons' output, controlled by a trainable parameter $\lambda$.
% Specifically, batch\_adj component takes two parameters ($t$ and $\theta$), and it outputs $\hat{z}$, as shown in Figure \ref{fig:method_qf_net}(c).
% The component indiv\_adj will finally determine the output of a neuron.
% It has only one parameter $\gamma$.
In the inference phase, its output $\tilde{z}$ can be calculated as follows.
\begin{equation}\label{equ:indi_adj}
\tilde{z}=\hat{z}\times(sin^2\frac{\gamma}{2})
\end{equation}
% As shown in Figure \ref{fig:method_qf_net}(c), this is similar to batch\_adj when $t=1$.
The determination of parameters $t$, $\theta$, and $\gamma$ is conducted in the training phase, which will be introduced later in \qfsim.

% With the neural computation and batch normalization components, a multi-layer neural network can be constructed, as shown in Figure \ref{fig:method_qf_net}(d).
% This process of neural computation and batch normalization will be repeated from the input layer until we reach the output layer's
% neuron.
% The outputs of these neurons will be converted to the classification outputs by a soft-max function on.
% The final classification procedure will be conducted on a classical computer.

% ; for instance, in Figure \ref{fig:method_qf_net}(d), if $O_0\ge O_1$, we classify the inputs to be the $1^{st}$ class, otherwise to be the $2^{nd}$ class.

% \subsection*{\revision{\qfhnet{t}}}

\subsection*{\qfsim}

\qfsim~involves both forward propagation and backward propagation. 
In forward propagation, all weights and parameters are determined, and we can conduct neural computation and batch normalization layer by layer.
\revision{For \qfnc$_{p}$, the neural computation will compute $y=\frac{\sum\nolimits_{\forall i}\{x_i\times w_i\}}{m}$ and  $y^2$, where $x_i$ is a two-point random variable.}
The distributions of $y$ and $y^2$ are illustrated in 
Figure \ref{fig:method_qf_sim}(a)-(b).
It is straightforward to get the expectation of $y^2$ by using the distribution; however, for $m$ inputs, it involves $2^m$ terms (e.g., $\prod q_i$ is one term), and leads to the time complexity to be $O(2^m)$. 
To reduce the time complexity, \qfsim~takes advantage of independence of inputs to calculate the expectation as follows:
\begin{equation}
\begin{aligned}
    E([\sum\nolimits_{\forall i}{w_ix_i}]^2)&=E(\sum\nolimits_{\forall i}[w_ix_i]^2+2\times\sum\nolimits_{\forall i}\sum\nolimits_{\forall j>i}[w_ix_iw_jx_j])\\
    &=m+2\times\sum\nolimits_{\forall i}\sum\nolimits_{\forall j>i}{E(w_ix_i)\times E(w_jx_j)}
\end{aligned}
\end{equation}
where $E(\sum\nolimits_{\forall i}[w_ix_i]^2)=m$, since $[w_ix_i]^2=1$ and there are $m$ inputs in total. 
The above formula derives the following algorithm with time complexity of $O(m^2)$ to simulate the neural computation {P}\qfnc{R}.

\begin{figure*}[t]
\centering
\includegraphics[width=1\linewidth]{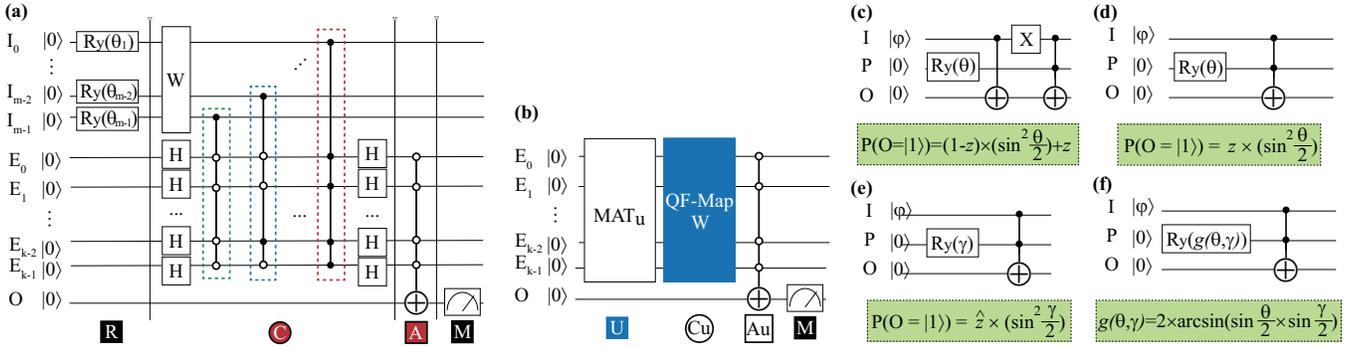}
\caption{\raggedr{\qfcirc: (a) quantum circuit designs for \qfnet; \revision{(b) quantum circuit design for \qfhnet{t}}; (c-f): batch normalization, quantum circuit designs for different cases; (c) design of ``batch\_adj'' for the case of $t=0$; (d) design of ``batch\_adj'' for the case of $t=1$; (e) design of ``indiv\_adj''; (f) optimized design for a specific case when $t=1$ in ``batch\_adj''.}}
\label{fig:method_qf_circ}
\end{figure*}

\noindent\begin{minipage}{0.49\textwidth}
\removelatexerror
\begin{algorithm}[H]
\footnotesize
\KwInput{(1) number of inputs $m$; (2) $m$ probabilities $\langle p_0,\cdots,p_{m-1}\rangle$; (3) $m$ weights $\langle w_0,\cdots,w_{m-1}\rangle$.}
\KwOutput{expectation of $y^2$}
 1. Expectation of random variable $x_i$: $e_i=E(x_i)=1-2\times p_i$;\\
 2. Expectation of $w_i\times x_i$: $E(w_i\times x_i)=w_i\times e_i$;\\
 3. Sum of pair product $sum_{pp} = \sum\nolimits_{\forall i}\sum\nolimits_{\forall j> i}\{E(w_i\times x_i)\times E(w_j\times x_j)\}$;\\
 4. Expectation of $y^2$: $E(y^2) = \frac{m+2\times sum_{pp}}{m^2}$;\\
 5. Return $E(y^2)$;
 \caption{\qfsim: simulating {P}\qfnc{R}}
\end{algorithm}
\end{minipage}
% In the above algorithm, computations in line 3 dominate the time of the whole procedure, which is composed of two loops with the size of $m$.
% Hence, the time complexity of the above algorithm is $O(m^2)$.
% , where $N=m$.
% It successfully decreases the time complexity from $O(2^m)$ to $O(m^2)$.
% Results in Table \ref{tab:mnist_simulation} show that the proposed \qfsim~is orders of magnitude faster than a quantum simulation on classical computers. 

\revision{For \qfnc$_{u}$, the neural computation will first convert inputs $I=\{i_0,i_1,\cdots,i_{m-1}\}$ to a vector $U=\{u_0,u_1,\cdots,u_{m-1}\}$ who can be the first column of a unitary matrix $MAT_u$. 
By operating $MAT_u$ on $K=log_2m$ qbits with initial state (i.e., $|0\rangle$), we can encode $U$ to $2^K=m$ states. 
The generating of unitary matrix $MAT_u$ is equivalent to the problem of identifying the nearest orthogonal matrix given a square matrix $A$.
Here, matrix $A$ is created by using $I$ as the first column, and $0$ for all other elements.
Then, we apply Singular Value Decomposition (SVD) to obtain $ B\sum C^*=SVD(A)$, and we can obtain $MAT_u=BC^*$.
Based on the obtained vector $U$ in $MAT_u$, \qfnc$_{u}$ computes $y=\frac{\sum\nolimits_{\forall i}\{u_i\times w_i\}}{m}$ and $y^2$, as shown in the following algorithm.
}

\noindent\begin{minipage}{0.49\textwidth}
\removelatexerror
\begin{algorithm}[H]
\footnotesize
\KwInput{(1) number of inputs $m$; (2) $m$ input values $\langle p_0,\cdots,p_{m-1}\rangle$; (3) $m$ weights $\langle w_0,\cdots,w_{m-1}\rangle$.}
\KwOutput{$[\frac{\sum\nolimits_{\forall i}\{u_i\times w_i\}}{m}]^2$}
 1. Generating square matrix A and compute $B\sum C^*=SVD(A)$\\
 2. Calculating $MAT_u=BC^*$ and extract vector $U$ from $MAT_u$ ;\\
 3. Compute $y=\frac{\sum\nolimits_{\forall i}\{u_i\times w_i\}}{m}$;\\
 4. Return $y^2$;
 \caption{\revision{\qfsim: simulating {U}\qfnc{R}}}
\end{algorithm}
\end{minipage}

The forward propagation for batch normalization can be efficiently implemented based on the output of the neural computation.
A code snippet is given as follows.

\noindent\begin{minipage}{0.49\textwidth}
\removelatexerror
\begin{algorithm}[H]
\footnotesize
\KwInput{(1) $E(y^2)$ from neural computation; (2) parameters $t$, $\theta$, $\gamma$ determined by training procedure.}
\KwOutput{normalized output $\tilde{z}$}
 1. Initialize $z$: $z=E(y^2)$;\\
 2. Calculate $\hat{z}$ according to Formula \ref{equ:batch_adj};\\
 3. Calculate $\tilde{z}$ according to Formula \ref{equ:indi_adj};\\
 4. Return $\tilde{z}$;
 \caption{\qfsim: simulating \qfbn}
\end{algorithm}
\end{minipage}

For the backward propagation, we need to determine weights and parameters (e.g., $\theta$ in \qfbn). 
The typically used optimization method (e.g., stochastic gradient descent \cite{bottou2010large}) is applied to determine weights.
In the following, we will discuss the determination of \qfbn parameters $t$, $\theta$, $\gamma$.

The batch\_adj sub-component involves two parameters, $t$ and $\theta$.
During the training phase, a batch of outputs are generated for each neuron.
Details are demonstrated in Figure \ref{fig:method_qf_sim}(c)-(d) with 6 outputs. In terms of the mean of outputs in a batch $p_{mean}$, there are two possible cases: (1)  $p_{mean}\le 0.5$ and (2) $p_{mean}>0.5$.
For the first case, $t$ is set to 0 and $\theta=2\times arcsin(\sqrt{\frac{0.5-p_{mean}}{1-p_{mean}}})$ can be derived from Formula \ref{equ:batch_adj} by setting $\hat{z}$ to 0.5; similarly, for the second case, $t$ is set to 1 and $\theta=2\times arcsin(\sqrt{\frac{0.5}{p_{mean}}})$.
% according to the calculation of $\hat{x}$, it is easy to derive that  $\theta=arccos(\frac{p_{mean}}{1-p_{mean}})$ (details see \todo{Appendix}).
% Similarly, when $p_{mean}>0.5$, $t$ is set to 1, and $\theta=arccos(\frac{p_{mean}-1}{p_{mean}})$.
Kindly note that the training procedure will be conducted in multiple iterations of batches.
% epochs (i.e., iterations).
As with the method for batch normalization in the conventional neural network, we employ \textit{moving average} to record parameters.
Let $x_i$ be the parameter of $x$ (e.g., $\theta$) at the $i^{th}$ iteration,
% epoch, 
and $x_{cur}$ be the value obtained in the current iteration.
For $x_{i}$, it can be calculated as $x_{i}=m\times x_{i-1}+ (1-m)\times x_{cur}$, where $m$ is the momentum which is set to $0.1$ by default in the experiments.

In forward propagation, the sub-module indiv\_adj is almost the same with batch\_adj for $t=0$; however, the determination of its parameter $\gamma$ is slightly different from $\theta$ for batch\_adj.
% is proposed to add varieties of outputs among different neurons.
% Likewise $\theta$ in batch\_adj for $t=0$, a parameter $\gamma$ is involved to adjust the angle.
% In addition, a new variable $\lambda$ is associated to each neuron, which can spread the outputs from 0 to 1.
% % Different values of $\lambda$ in neurons will lead to their outputs to be distinct.
% % The component indiv\_adj has two parameters, $\gamma$ and $\lambda$.
% % $\lambda$ is associated
% % In the inference phase, its function is similar to batch\_adj when $t=1$.
% % However, the determination of $\gamma$ is slightly different from that of $\theta$.
% The training procedure is discussed as follow.
As shown in Figure \ref{fig:method_qf_sim}(e), the initial probability of $\hat{z}$ after batch\_adj is $p_z$.
The basic idea of indiv\_adj is to move $\hat{z}$ by an angle, $\gamma$.
It will be conducted in three steps: (1) we move start point at $p_z$ to point $A$ with the probability of $(p_z/n+0.5)\times \lambda$, where $n$ is the batch size and $\lambda$ is a trainable variable; (2) we obtain $\gamma$ by moving point $A$ to $p=0.5$; (3) we finally move solution at $p_z$ by the angle of $\gamma$ to obtain the final result.
By replacing $P_{mean}$ by $(p_z/n+0.5)\times \lambda$ in batch\_adj when $t=1$, we can calculate $\gamma$.
% By replacing $P_{mean}$ 
% and $\gamma$ is calculated as $\gamma=arccos(\frac{(p_{x}/n+0.5)\times \lambda-1}{(p_{x}/n+0.5)\times \lambda})$, where $\lambda$ is a trainable variable.
For each batch, we calculate the mean of $\gamma$, and we also employ the \textit{moving average} to record $\gamma$.

\subsection*{\qfcirc}
% how to convert our 
% \qfnet~into its quantum circuit implementation.

We now discuss the corresponding circuit design for components in \qfsnet{s}, including {P}\qfnc{R}, {U}\qfnc{R}, and \qfbn.
Figures \ref{fig:method_qf_circ}(a)-(b) demonstrate the circuit design for {P}\qfnc{R} (see Figure \ref{fig:method_qf_net}(c)) and {U}\qfnc{R} (see Figure \ref{fig:method_qf_net}(e)), respectively; Figures \ref{fig:method_qf_circ}(c)-(f) demonstrate the \qfbn~in Figure \ref{fig:method_qf_bn}.
% The quantum circuit of the neural computation in Figure \ref{fig:method_qf_net}(c) is illustrated in Figure \ref{fig:method_qf_circ}(a); while the design of different cases of batch normalization in Figure \ref{fig:method_qf_net}(b) are illustrated in Figure \ref{fig:method_qf_circ}(b)-(e).
% A detailed demonstration of the equivalency between \qfcirc~and {P}\qfnc{R}~can be found in the \textbf{Supplementary Information}.

\revision{\textbf{Implementing {P}\qfnc{R} on quantum circuit:}}
For an $m$-input neural computation, the quantum circuit for {P}\qfnc{R} is composed of $m$ input qbits (I), and $k=log_2m$ encoding qbits (E), and 1 output qbit (O).

\revision{In accordance with the operations in {P}\qfnc{R}, the circuit is composed of four parts.
In the first part, the circuit is initialized to perform \underline{\textbf{R}} operation.
For qbits $I$, we apply $m$ Ry gates with parameter $\theta=2\times arcsin(\sqrt{p_k})$ to initialize the input qbit $I_k$ in terms of the input real value $p_k$, such that the state of $I_k$ is changed from $|0\rangle$ to $\sqrt{q_k}|0\rangle+\sqrt{p_k}|1\rangle$.
For encoding qbits $E$ and output qbit $O$, they are initialized as $|0\rangle$.
The second part completes the average sum function, i.e., \underline{\textbf{C}} operation. 
It further includes three steps: (1) dot product of inputs and weights on qibits $I$, (2) make encoding qbits $E$ into superposition, (3) encode $m$ probabilities in qbits $I$ to $2^k=m$ states in qbits $E$.
The third part implements the quadratic activation function, that is the \underline{\textbf{A}} operation. It applies the control gate to extract the amplitudes in states $|I_0I_1\cdots I_{m-1}\rangle\otimes|00\cdots0\rangle$ to qbit $O$.
As we know that the probability is the square of the amplitude, the quadratic activation function can be naturally implemented.
Finally, \underline{\textbf{E}} operations  corresponds to the fourth part that measures qbit $O$ to obtain the output real number $E(y^2)$, where the state of $O$ is $|O\rangle=\sqrt{1-E(y^2)}|0\rangle+\sqrt{E(y^2)}|1\rangle$.
A detailed demonstration of the equivalency between \qfcirc~and {P}\qfnc{R}~can be found in the \textbf{Supplementary Information}.}

Kindly note that for a multi-layer network composed of {P}\qfnc{R}, namely \qfnet, there is no need to have a measurement at interfaces, because the converting operation $R$ initializes a qbit to the state exactly the same with $|O\rangle$.
In addition, the batch normalization can also take $|O\rangle$ as input.

\revision{\textbf{Implementing {U}\qfnc{R} on quantum circuit:}}
\revision{
For an $m$-input neural computation, the quantum circuit for {U}\qfnc{R} contains $k=log_2m$ encoding qbits $E$ and 1 output qbit $O$.}

% In accordance with the operations in
\revision{According to {U}\qfnc{R}, the circuit in turn performs \underline{\textbf{U}}, \underline{\textbf{$C_u$}}, \underline{\textbf{$A_u$}} operations, and finally obtains the result by a measurement.
In the first operation, unlike the circuit for {P}\qfnc{R} using $R$ gate to initialize circuits using $m$ qbits; for {U}\qfnc{R}, we using the matrix $MAT_u$ to initialize circuits on $k=log_2m$ qbits $E$.
Recalling that the first column of $MAT_u$ is vector $V$, after this step, $m$ elements in vector $V$ will be encoded to $2^k=m$ states represented by qbits $E$.
The second operation is to perform the dot product between all states in qbits $E$ and weights $W$, which is implemented by control Z gates and will be introduced in \qfmap.
Finally, like the circuit for {P}\qfnc{R}, the quadratic activation and measurement are implemented.
Kindly note that, in addition to quadratic activation, we can also implement higher orders of non-linearity by duplicating the circuit to perform \underline{\textbf{U}}, \underline{\textbf{$C_u$}}, and \underline{\textbf{$A_u$}} to achieve multiple outputs.
Then, we can use control NOT gate on the outputs to achieve higher orders of non-linearity.
For example, using a Toffoli gate on two outputs can realize $y^4$.
Let the non-linear function be $y^k$ and the cost complexity of {U}\qfnc{R} using quadratic activation be $O(N)$, then the cost complexity of {U}\qfnc{R} using $y^k$ as the non-linear function will be $O(kN)$.}

\revision{For neural networks with given inputs, we can preprocess the \underline{\textbf{U}} operation and store the states in quantum memory \cite{lvovsky2009optical}.
Thus, the key for quantum advantage is to exponentially reduce the number of gates used in neural computation, compared with the number of basic operators used in classical computing.
We will present an algorithm in \qfmap~for {U}\qfnc{R} to achieve this goal.}

\revision{\textbf{Implementing \qfbn~on quantum circuit:}}
Now, we discuss the implementation of \qfbn~in quantum circuits.
In these circuits, three qbits are involved: (1) qbit $I$ for input, which can be the output of qbit $O$ in circuit without measurement, or initialized using a Ry gate according to the measurement of qbit $O$ in circuit; (2) qbit $P$ conveys the parameter, which is obtained via training procedure, see details in \qfsim; (3) output qbits $O$, which can be directly used for the next layer or be measured to convert to a real number.

Figures \ref{fig:method_qf_circ}(b)-(c) show the circuit design for two cases in batch\_adj.
Since parameters in batch\_adj are determined in the inference phase, if $t=0$, we will adopt the circuit in Figure \ref{fig:method_qf_circ}(b), otherwise, we adopt that in Figure \ref{fig:method_qf_circ}(c).
Then, Figure \ref{fig:method_qf_circ}(d) shows the circuit for indiv\_adj.
We can see that circuits in Figures \ref{fig:method_qf_circ}(c) and (d) are the same except the initialization of parameters, $\theta$ and $\gamma$.
For circuit optimization, we can merge the above two circuits into one by changing the input parameters to $g(\theta,\gamma)$, as shown in Figure \ref{fig:method_qf_circ}(e).
In this circuit, $\tilde{z}^{\prime}=z\times sin^2\frac{g(\theta,\gamma)}{2}$, while for applying circuits in Figures \ref{fig:method_qf_circ}(c) and (d), we will have $\tilde{z}=z\times sin^2\frac{\theta}{2}\times sin^2\frac{\gamma}{2}$.
To guarantee the consistent function, we can derive that $g(\theta,\gamma)=2\times arcsin(sin\frac{\theta}{2}\times sin\frac{\gamma}{2})$.

\subsection*{\qfmap}

\qfmap~is an automatic tool to map \qfsnet{s}~to the quantum processor through two steps: network-to-circuit mapping, which maps \qfsnet{s} to \qfcirc; and virtual-to-physic mapping, which maps \qfcirc~to physic qbits.

\revision{\textbf{Mapping \qfsnet{s} to \qfcirc:}
The first step of \qfmap~is to map three kinds of layers (i.e., {P}\qfnc{R}, {U}\qfnc{R}, and \qfbn) to \qfcirc.
The mappings of {P}\qfnc{R} and \qfbn~are straightforward.
Specifically, for {P}\qfnc{R} in Figure \ref{fig:method_qf_circ}(a), the circuit for weight $W$ is determined using the following rule: for a qbit $I_k$, an $X$ gate is placed if and only if $W_k=-1$.
Let the probability $P(x_k=-1)=P(I_k=|0\rangle)=q_k$, after the $X$ gate, the probability becomes $1-q_k$.
Since the values of random variable $x_k$ are $-1$ and $+1$, such an operation computes $-x_k$.
For \qfbn, let $O_1$ be the output qbit of the first layer. It can be directly connected to the qbit $I$ in Figure \ref{fig:method_qf_circ}(c)-(f), according to the type of batch normalization, which is determined by the training phase.}

\revision{
The mapping of {U}\qfnc{R} to quantum circuits is the key to achieve quantum advantages. 
In the following texts, we will first formulate the problem, and then introduce the proposed algorithm to guarantee the cost for a neural computation with $2^k$ inputs to be $O(k^2)$.}

\revision{
Before formally introducing the problem, we first give some fundamental definitions that will be used.
We first define the quantum state and the relationship between states as follows.
Let the computational basis be composed of $k$ qbits, as in Figure \ref{fig:method_qf_circ}(b).
Define $|x_i\rangle=|b^i_{k-1},\cdots,b^i_j,\cdots,b^i_{0}\rangle$ to be the $i^{th}$ state, where $b_j$ is a binary number and $x_i=\sum\nolimits_{\forall j}\{b^i_j\cdot 2^{j}\}$.
For two states $|x_p\rangle$ and $|x_q\rangle$, we define $|x_p\rangle\subseteq |x_q\rangle$ if $\forall b^p_j=1$, we have $b^q_j=1$. 
We define $sign(x_i)$ to be the sign of $x_i$. 
}

\begin{figure}[t]
\centering
\includegraphics[width=0.9\linewidth]{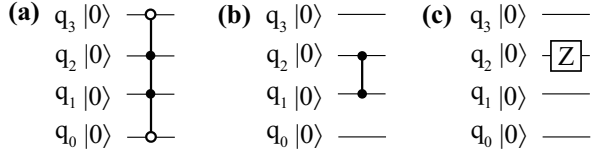}
\caption{\raggedr{\revision{Illustration of state $|6\rangle=|0110\rangle$ and $|4\rangle=|0100\rangle$ in a $k=4$ computation system: (a) $FG_6$; (b) $PG_6$; (c) $PG_4$.
}}}
\label{fig:gates}
\end{figure}

\revision{
Next, we define the gates to flip the sign of states.
The controlled $Z$ operation among $K$ qbits (e.g., $C^KZ$) is a quantum gate to flip the sign of states \cite{nielsen2002quantum,tacchino2019artificial}.
Define $FG_{x_i}$ to be a $C^KZ$ gate to flip the state $x_i$ only.
It can be implemented as follows: if $b^i_j=1$, the control signal of the $j^{th}$ qbit is enabled by $|1\rangle$, otherwise if $b^i_j=0$, it is enabled by $|0\rangle$.
Define $PG_{x_i}$ to be a controlled Z gate to flip all states $\forall x_m$ if $x_m\subseteq x_i$.
It can be implemented as follows: if $b^i_j=1$, there is a control signal of the $j^{th}$ qbit, enabled by $|1\rangle$, otherwise, it is not a control qbit.
Specifically, if there is only $b^i_{m}=1$ for all $b^i_k\in x_i$, we put a $Z$ gate on the $m^{th}$ qbit. 
Figure \ref{fig:gates} illustrates $FG_6$, $PG_6$ and $PG_4$.
% Different control Z gates have different costs.
We define cost function $\mathbb{C}$ to be the number of basic gates (e.g., Pauli-X, Toffoli) used in a control Z gate.
}

\revision{Now, we formally define the weight mapping problem as follows: Given (1) a vector $W=\{w_{m-1},\cdots,w_0\}$ with $m$ binary weights (i.e., -1 or +1) and (2) a computational basis of $k=log_2m$ qbits that include $m$ states $X=\{x_{m-1},\cdots,x_0\}$ and $\forall x_j\in X$, $sign(x_j)$ is $+$, the problem is to determine a set of gates $G$ in either $FG$ or $PG$ to be applied, such that the circuit cost 
% of all involved gates 
is minimized, while the sign of each state is the same with the corresponding weight; i.e., $\forall x_j\in X$, $sign_G(x_j)=sign(w_j)$ and $min=\sum\nolimits_{g\in G}(\mathbb{C}(g))$, where $sign_G(x_j)$ is the sign of state $x_j$ under the computing conducted by a sequence of quantum gates in $G$.}
% Since the weights are either $+1$ or $-1$, the $W$ component in Figure \ref{fig:method_qf_circ}(b) is to flip the sign of all states formed by $k$ qbits, in accordance to the trained weights.

\revision{A straightforward way to satisfy the sign flip requirement without considering cost is to apply $FG$ for all states whose corresponding weights are $-1$.
A better solution for cost minimization is to use hypergraph states\cite{tacchino2019artificial}, which starts from the states with less $|1\rangle$, and apply $PG$ to reduce the cost.
However, as shown in the previous work, both methods have the cost complexity of $O(2^k)$, which is the same as classical computers and no quantum advantage can be achieved.}

\revision{Toward the quantum advantage, we made the following important observation: the order of weights can be adjusted, since matrix $MAT_u^\prime$ obtained by switching two rows in the unitary matrix $MAT_u$ will still be a unitary matrix. Based on this property, we can simplify the weight mapping problem to determine a set of gates, such that the cost is minimized while the number of states with sign flip is the same as the number of $-1$ in weight $W$.
On top of this, we propose an algorithm to guarantee the cost complexity to be $O(k^2)$.
Compared to $O(2^k)$ needed for classical computers, we can achieve quantum advantages.
To demonstrate how to guarantee the cost complexity to be $O(k^2)$, we first have the following theorem.}

% \begin{theorem}
% \label{thm:cover}
% \revision{For an integer number $R$ where $R>0$ and $R\le2^{k-1}$, it can be expressed by a set of numbers $S=\{2^{i}|0\le i<k\}$ and signs $\{-,+\}$ with the following rules:
% (1) A number in set $S$ can be selected for only once or not selected;
% (2) for all the selected numbers, the sign of the maximum one is $+$;
% (3) all the selected numbers are sorted in the descending order, and their signs are alternatively changed.}
% % it can be achieved by only using the gates in the set $\{PG_{x(j)}\ |\ x(j)=2^j-1\ and\ j\in[1,k]\}$.}
% \end{theorem}

\begin{theorem}
\label{thm:cover}
\revision{For an integer number $R$ where $R>0$ and $R\le2^{k-1}$,
the number $R$ can be expressed by a sequence of addition ($+$) or subtraction ($-$) of a subset of $S_k=\{2^{i}|0\le i<k\}$; when the terms of $S_k$ are sorted in a descending order, the sign of the expression (addition and subtraction) are alternative with the leading sign being addition ($+$).
% it can be expressed by a subset $Q$ of numbers in the set $S=\{2^{i}|0\le i<k\}$ and signs $\{-,+\}$ with the following rules:
% (1) A number in set $S$ can be selected for only once or not selected;
% (2) for all the selected numbers, the sign of the maximum one is $+$;
% (3) all the selected numbers are sorted in the descending order, and their signs are alternatively changed.
}
% it can be achieved by only using the gates in the set $\{PG_{x(j)}\ |\ x(j)=2^j-1\ and\ j\in[1,k]\}$.}
\end{theorem}

\begin{proof}
\revision{
The above theorem can be proved by induction.
First, for $k=2$, the possible values of $R$ are $\{1,2\}$, the set $S_2=\{2,1\}$. The theorem is obviously true.
Second, for $k=3$, the possible values of $R$ are $\{1,2,3,4\}$, and the set $S_3=\{4,2,1\}$.
In this case, only $R=3$ needs to involve 2 numbers from $S_3$ using the expression $3=4-1$; other numbers can be directly expressed by themselves.
So, the theorem is true.}

\revision{Third, assuming the theorem is true for $k=n-1$, we can prove that for $k=n$ the theorem is true for the following three cases.
Case 1: For $R<2^{n-2}$, since the theorem is true for $k=n-1$, based on the assumption, all numbers less than $2^{n-2}$ can be expressed by using set $S_{n-1}$ and thus we can also express them by using set $S_n$ because $S_{n-1}\subseteq S_n$;
Case 2: For $R=2^{n-2}$, itself is in set $S_{n}$;
Case 3: For $R>2^{n-2}$, we can express $R=2^{n-1}-T$, where $T=2^{n-1}-R< 2^{n-2}$. Since the theorem is true for $k=n-1$, we can express $T$ by using set $S_n-{2^{n-1}}=\{2^{i}|0\le i<n-1\}=S_{n-1}$; hence, $R$ can be expressed using set $S_n$.}

\revision{Above all, the theorem is correct.}
% First, if $R>2^{k-1}$, we can apply the symmetric property and change the problem to find solution with $2^k-R$ sign flips, which is less than $2^{k-1}$. 
% Second, if $R\le 2^{k-1}$, we can prove the theorem by induction.
% It is obvious that for the gate $PG_x$ where $x=2^j-1$, it will flip the signs of $2^{k-j}$ states. In addition, $\forall j<k$ $PG_{x(k)}\subseteq PG_{x(j)}$.
% For $k=2$, we have set $\{PG_{x(1),PG_{x(2)}}\}$, whose sign flip number are 2 and 1, which covers all 
\end{proof}

\revision{We propose to only use $PG$ gate in a set $\{PG_{x(j)}\ |\ x(j)=2^j-1\ \&\ j\in[1,k]\}$ for any required number $R\in(0,2^{k})$ of sign flips on states; for instance, if $k=4$, the gate set is $\{PG_{0001}, PG_{0011}, PG_{0111}, PG_{1111}\}$ and $R\in(0,16)$.
This can be demonstrated using the above theorem and properties of the problem: (1) the problem has symmetric property due to the quadratic activation function. Therefore, the weight mapping problem can be reduced to find a set of gates leading to the number of $-1$ no larger than $2^{k-1}$; i.e., $R\in(0,2^{k-1}]$. (2) for $PG_{x(j)}$, it will flip the sign of $2^{k-j}$ states; since $j\in[1,k]$, the numbers of the flipped sign by these gates belong to a set $S_k=\{2^i|0\le i<k\}$;
These two properties make the problem in accordance with that in Theorem \ref{thm:cover}.
The weight mapping problem is also consistent with three rules in the theorem.
(1) A gate can be selected or not, indicating the finally determined gate set is the subset of $S_k$; (2) all states at the beginning have the positive sign, and therefore, the first gate will increase sign flips, indicating the leading sign is addition ($+$); (3) $\forall p<q$, $x(q)\subseteq x(p)$; it indicates that among the $2^{k-p}$ states whose signs are flipped by $x(p)$, there are $2^{k-q}$ states signs are flipped back; this is in accordance to alternatively use $+$ and $-$ in the expression in Theorem \ref{thm:cover}.
Followed by the proof procedure, we devise the following recursive algorithm to decide which gates to be employed.
}

\noindent\begin{minipage}{0.49\textwidth}
\removelatexerror
\begin{algorithm}[H]
\footnotesize
\KwInput{(1) An integer $R\in(0,2^{k-1}]$; (2) number of qbits $k$;}
\KwOutput{A set of applied gate $G$}
  void recursive($G$,$R$,$k$)\{\\
    \ \ \ \ \ \ \ \  if $(R<2^{k-2})$\{\\
    \ \ \ \ \ \ \ \ \ \ \ \ \ \ \ \  recursive($G$,$R$,$k-1$); // Case 1 in the third step\\
    \ \ \ \ \ \ \ \  \}\\
    \ \ \ \ \ \ \ \  else if $(R==2^{k-1})$\{\\
    \ \ \ \ \ \ \ \ \ \ \ \ \ \ \ \  $G.append(PG_{2^{k-1}})$; // Case 2 in the third step\\
    \ \ \ \ \ \ \ \ \ \ \ \ \ \ \ \  return;\\
    \ \ \ \ \ \ \ \  \}else\{\\
    \ \ \ \ \ \ \ \ \ \ \ \ \ \ \ \  $G.append(PG_{2^{k-1}})$;\\
    \ \ \ \ \ \ \ \ \ \ \ \ \ \ \ \  recursive($G$,$2^{k-1}-R$,$k-1$); // Case 3 in the third step\\
    \ \ \ \ \ \ \ \  \}\\
  \} \\
  // Entry of weight mapping algorithm\\
  \textit{set} main($R$,$k$)\{\\
    \ \ \ \ \ \ \ \  Initialize empty \textit{set} $G$;\\
    \ \ \ \ \ \ \ \  recursive($G$,$R$,$k$);\\
    \ \ \ \ \ \ \ \  return $G$\\
  \}
%  1. Initialize $z$: $z=E(y^2)$;\\
%  2. Calculate $\hat{z}$ according to Formula \ref{equ:batch_adj};\\
%  3. Calculate $\tilde{z}$ according to Formula \ref{equ:indi_adj};\\
%  4. Return $\tilde{z}$;
 \caption{\revision{\qfmap: weight mapping algorithm}}
\end{algorithm}
\end{minipage}

\revision{In the above algorithm, the worst case for the cost is that we apply all gates in $\{PG_{x(j)}\ |\ x(j)=2^j-1\ \&\ j\in[1,k]\}$.
Let the state $x(j)$ has $y$ $|1>$ states, if $y>2$ the $PG_{x(j)}$ can be implemented using $2y-1$ basic gates, including $y-1$ Toffoli gates for controlling, 1 control Z gate, and $y-1$ Toffoli gates for resetting; otherwise, it uses 1 basic gates.
Based on these understandings, we can calculate the cost complexity in the worst case, which is $1+1+3+\cdots+(2\times k-1)=k^2+1$.
Therefore, the cost complexity of linear function computation is $O(k^2)$.
The quadratic activation function is implemented by a $C^kZ$ gate, whose cost is $O(k)$.
Thus, the cost complexity for neural computation {U}\qfnc{R} is $O(k^2)$.}

\revision{Finally, to make the functional correctness, in generating the inputs unitary matrix, we swap rows in it in terms of the weights, and store the generated results in quantum memory.}

% for the former gates that flip }

% \revision{Second, we utilize the following symmetric property of weight mapping problem to reduce the problem to find a set of gates leading to the number of $-1$ less than $2^{k-1}$.
% Owning to the quadratic activation function, if the signs of all weights in $W$ are flipped, the results are the same.}

% \revision{Third,
% The property of these states is $\forall p<q$, $x(q)\subseteq x(p)$.
% In addition, for $PG_{x(j)}$, it will flip the sign of $2^{k-j}$ states.
% Therefore, $$

% In addition, these states have .
% (e.g., states 0001, 0011, 0111, 1111 for $k=4$).
% For $PG_x$, it will flip the sign of $2^{k-j}$ states (e.g., $PG_1$ will flip the sign of $2^{4-1}=8$ states, including \{0001, 0011, 0101, 0111, 1001, 1011, 1101, 1111\}).
% }

% Therefore, if the required number $R$ of sign flip is larger than $2^{k-1}$, we can use the above property to find solution with $2^k-R$ sign flips, which is less than $2^{k-1}$. 
% We have the following theorem.}

% We observe that for any number $SF\le2^{k-1}$, we can only use gate $PG_X$, where $x=2^j-1$, $j\in\{1,2,\cdots,k\}$ (e.g., states 0001, 0011, 0111, 1111 for $k=4$).
% For $PG_x$, it will flip the sign of $2^{k-j}$ states (e.g., $PG_1$ will flip the sign of $2^{4-1}=8$ states, including \{0001, 0011, 0101, 0111, 1001, 1011, 1101, 1111\}).

% In addition, as stated in work\cite{tacchino2019artificial}, the problem is symmetric, that is, we can flip all sign of weights in $W$.

\textbf{Mapping \qfcirc~to physic qbits:}
After \qfcirc~is generated, the second step is to map \qfcirc~to quantum processors, called virtual-to-physic mapping. 
In this paper, we deploy \qfcirc~to various IBM quantum  processors.
Virtual-to-physic mapping in \qfmap~has two tasks: (1) select a suitable quantum processor backend, and (2) map  qbits in \qfsnet{s}~to physic qbits in the selected backend.
For the first task, \qfmap~will \textit{i)} check the number of qbits needed; \textit{ii)} find the backend with the smallest number of qbit to accommodate \qfcirc; \textit{iii)} for the backends with the same number of qbits, \qfmap~will select a backend for the minimum average error rate.
The second task in \qfmap~is to map qbits in \qfsnet{s}~to physic qbits. 
The mapping follows two rules: (1) the qbit in \qfsnet{s}~with more gates is mapped to the physic qbit with a lower error rate; and (2)  qbits in \qfsnet{s}~with connections are mapped to the physic qbits with the smallest distance.

\section*{Data availability}
The authors declare that all data supporting the findings of this study are available within the article and its Supplementary Information files. Source data can be accessed via \url{https://wjiang.nd.edu/categories/qf/}.

\section*{Code availability}
All relevant codes will be open in github upon the acceptance of the manuscript or be available from the corresponding authors upon reasonable request.

\small{
\bibliography{sample}
}
% \noindent LaTeX formats citations and references automatically using the bibliography records in your .bib file, which you can edit via the project menu. Use the cite command for an inline citation, e.g.  \cite{Hao:gidmaps:2014}.

% For data citations of datasets uploaded to e.g. \emph{figshare}, please use the \verb|howpublished| option in the bib entry to specify the platform and the link, as in the \verb|Hao:gidmaps:2014| example in the sample bibliography file.

\section*{Acknowledgements}
This work is partially supported by IBM
and University of Notre Dame (IBM-ND) Quantum program.

% Acknowledgements should be brief, and should not include thanks to anonymous referees and editors, or effusive comments. Grant or contribution numbers may be acknowledged.

\section*{Author contributions statement}
W.J. conceived the idea and performed quantum evaluations; J.X. and Y.S. supervised the work and improved the idea and experiment design. All authors contributed to manuscript writing and discussions about the results.

\clearpage

\end{document}